\documentclass[journal]{IEEEtran}

\usepackage{amsmath,amssymb,amsfonts}
\usepackage{graphicx}
\usepackage{textcomp}
\usepackage{booktabs}
\usepackage{multirow}
\usepackage{float}

\newtheorem{theorem}{Theorem}
\newtheorem{proposition}{Proposition}
\newtheorem{lemma}{Lemma}
\newtheorem{corollary}{Corollary}
\newtheorem{remark}{Remark}
\newtheorem{assumption}{Assumption}

\newcommand{\R}{\mathbb{R}}
\newcommand{\cC}{\mathcal{C}}

\newcommand{\cH}{\mathcal{H}}

\newcommand{\cK}{\mathcal{K}}
\newcommand{\cU}{\mathcal{U}}

\newcommand{\norm}[1]{\left\lVert #1 \right\rVert}
\newcommand{\abs}[1]{\left\lvert #1 \right\rvert}
\newcommand{\trans}{\mathsf{T}}

\title{Robust Safety Filtering for Input-Constrained Underactuated Linear Systems}

\author{Muhamad Rausyan Fikri%
\thanks{Muhamad Rausyan Fikri is with the Faculty of Engineering and Natural Sciences, Tampere University, 33720 Tampere, Finland (e-mail: muhamad.fikri@tuni.fi).}%
}

\begin{document}

\maketitle

\begin{abstract}
We develop a robust safety filter for input-constrained underactuated linear systems subject to unknown bounded-rate disturbances and structured model uncertainty entering through known distribution channels.
A disturbance observer provides an online disturbance estimate and a dynamic estimation-error radius that define the certified uncertainty envelope used in the robust high-order control barrier function constraints.
For general polytopic actuator sets, the resulting robust safety-admissible input set is polyhedral, while the shared scalar-input specialization yields an exact interval, a necessary-and-sufficient pointwise feasibility condition, a signed feasibility reserve, and a closed-form safety projection.
A domain-wise certificate relates robust HOCBF control demand to available actuator authority, while a finite-horizon energy identity quantifies deviation from the unconstrained $H_\infty$ reference without claiming preservation of its original attenuation level.
Across all $16$ structured-uncertainty corners in the linear-model actuator-stress test, the filter remains feasible and safe under $\tau_{\max}=1.18$ N$\cdot$m, with $\max\abs{p}=0.0740$ m, $\max\abs{\theta}=0.205$ rad, $\max\abs{\tau}=1.173$ N$\cdot$m, and $\min\mu_S=7.27\times10^{-3}$ N$\cdot$m.
\end{abstract}

\begin{IEEEkeywords}
Disturbance observers, differential games, $H_\infty$ control, high-order control barrier functions, underactuated systems.
\end{IEEEkeywords}

\section{Introduction}

Autonomous systems operating under unknown disturbances must maintain closed-loop performance while satisfying hard state and actuator constraints.
LQR provides nominal quadratic regulation, whereas $H_\infty$ control limits the worst-case disturbance-to-output $\mathcal{L}_2$ gain \cite{basar1995}; neither formulation, by itself, enforces hard state-safety constraints.
This issue is particularly relevant to underactuated mechanical systems, where multiple safety requirements can compete through fewer independent actuation channels.
Underactuation is not required by the scalar-input feasibility result developed here; rather, it provides a natural setting in which several safety constraints act through the same bounded control authority.

Control barrier functions (CBFs) impose online input constraints that render prescribed safe sets forward invariant \cite{ames2017}.
For higher-relative-degree safety functions, high-order CBFs (HOCBFs) recursively construct input-dependent conditions, including safety filters for underactuated mechanical systems \cite{xiao2022,wu2025underactuated}.
Robust and observer-based CBF/HOCBF methods further incorporate bounded uncertainty, disturbance estimates, or uncertainty estimates into certified safety constraints \cite{jankovic2018,kolathaya2019,breeden2023,alan2023,wangxu2023,sun2024,dasburdick2025}.

Feasibility and compatibility have also been studied when multiple barrier constraints compete for limited control authority \cite{xiao2022feasibility,zhang2024feasibility}.
Jiang et al.~\cite{jiang2026compatibility} derive necessary-and-sufficient compatibility conditions for CLF and multiple HOCBF constraints with and without input limits, together with SOS-based verification and synthesis.
Accordingly, generic compatibility of affine barrier inequalities under bounded inputs is not claimed as a contribution of the present work.
Dong et al.~\cite{dong2026vcbf} combine disturbance-observer information with HOCBF and input constraints and introduce a volume CBF to prevent collapse of the QP feasible region under uncertainty.
Wei et al.~\cite{wei2026explicit} derive explicit HOCBF-based controllers for high-relative-degree systems subject to disturbances.
Cohen et al.~\cite{cohen2025compatibility} establish compatibility and a closed-form optimization-based controller for multiple CBFs induced by box-constrained outputs under a vector-relative-degree condition, while their formulation does not incorporate hard input bounds.

The present work instead focuses on how dynamically certified disturbance information and structured model uncertainty modify the safety-admissible control geometry under hard actuator limits.
A disturbance observer generates both an online disturbance estimate and a dynamic estimation-error radius, which together define a certified uncertainty envelope used to construct the exact robust HOCBF counterpart.
For general polytopic input constraints, the resulting inequalities define a robust safety-admissible polytope whose nonemptiness can be checked through finite-dimensional linear feasibility.
For a shared bounded scalar input, this polytope reduces to an interval whose endpoints provide the signed feasibility reserve $\mu_S$ and the corresponding closed-form projection.
An additional domain-wise certificate relates robust HOCBF control demand to available actuator authority before the online pointwise test is evaluated.
This exactness refers to the certified disturbance and model-uncertainty envelope, which may overapproximate the true dynamically coupled uncertainty set.
The baseline $H_\infty$ controller supplies the performance reference used subsequently for finite-horizon deviation accounting.

The main contributions are as follows.
\begin{enumerate}
\item \emph{Observer-certified robust input feasibility:}
We combine a dynamic disturbance-estimation certificate with structured model uncertainty to obtain robust HOCBF constraints under hard actuator limits.
The resulting set is polyhedral for general polytopic inputs, while the shared scalar-input case yields exact endpoints, a necessary-and-sufficient pointwise feasibility condition, the signed reserve $\mu_S$, and a closed-form projection.

\item \emph{A priori actuator-authority certification and robust safety:}
We derive a domain-wise sufficient condition that lower-bounds the shared-input feasibility reserve and relates robust HOCBF control demand to available actuator authority.
The certificate depends explicitly on the disturbance-rate bound, observer gain, model-uncertainty bound, and actuator limits, while forward invariance follows whenever closed-loop feasibility is maintained.

\item \emph{Finite-horizon $H_\infty$ baseline-deviation accounting:}
We derive an uncertainty-aware finite-horizon identity and bound that separate departure from the unconstrained $H_\infty$ reference from the contribution of structured model uncertainty.
The result quantifies the effect of disturbance compensation and safety filtering without claiming preservation of the original attenuation level.
\end{enumerate}
The baseline $H_\infty$ synthesis, disturbance-observer structure, polyhedral feasibility test, and HOCBF recursion are supporting rather than individually claimed contributions.
Numerical validation on a two-wheeled balancing robot includes architecture ablations, structured-uncertainty and actuator-stress tests, and an out-of-model nonlinear study without controller retuning.

The remainder of this paper is organized as follows.
Section~II introduces the system setting and HOCBF preliminaries.
Section~III develops the robust safety filter and its feasibility, safety, stability, and finite-horizon properties.
Section~IV presents the TWBR study, followed by the conclusion in Section~V.

\section{Preliminaries}
\label{sec:preliminaries}

\noindent\textit{Notation:}
The sets of real, positive real, and nonnegative real numbers are denoted by $\R$, $\R_{+}$, and $\R_{\geq0}$, respectively.
The notation $\norm{\cdot}$ denotes the Euclidean norm for vectors and the induced Euclidean norm for matrices.
For a positive definite matrix $R\succ0$, define $\norm{v}_R^2:=v^\trans Rv$.
For a measurable signal $w:[0,\infty)\rightarrow\R^r$, define $\norm{w}_{\mathcal L_2}:=\left(\int_0^\infty\norm{w(t)}^2dt\right)^{1/2}$ whenever the integral is finite.
A zero vector of appropriate dimension is denoted by $0$.

\subsection{System Class and Shared-Actuation Structure}

Consider the linear control system
\begin{equation}
    \dot x=Ax+Bu,
    \label{eq:linear-system-prelim}
\end{equation}
where $x\in\R^n$, $u\in\cU\subseteq\R^m$, and $\cU$ is a compact convex polytope.
The disturbed and uncertain extension of \eqref{eq:linear-system-prelim} is introduced in Section~\ref{sec:proposed-framework}.

To relate \eqref{eq:linear-system-prelim} to underactuated mechanical systems, consider generalized coordinates $q\in\R^{n_q}$, an inertia matrix $M(q)\succ0$, and an actuation distribution $B_q(q)$ with $\operatorname{rank}B_q(q)<n_q$.
Around an equilibrium $q_{\mathrm e}$, let $x=[q^\trans\ \dot q^\trans]^\trans$, for which the linearized input matrix is $B=[0^\trans\ (M^{-1}(q_{\mathrm e})B_q(q_{\mathrm e}))^\trans]^\trans$.
Since $M(q_{\mathrm e})$ is nonsingular, $\operatorname{rank}B=\operatorname{rank}B_q(q_{\mathrm e})$, so the local actuation deficiency is preserved.

The role of underactuation is structural rather than theorem-defining because multiple safety requirements may compete for fewer independent actuation channels.
For the scalar-input specialization, $m=1$ and $\cU=[u_{\min},u_{\max}]$ with $u_{\min}<u_{\max}$, so affine HOCBF inequalities constrain the same scalar input and their intersection remains an interval.
The present derivation is restricted to relative-degree-two barriers because higher-order robust HOCBF recursions would generally require certified propagation of additional uncertainty derivatives.
Thus, the relative-degree-two restriction concerns uncertainty certification rather than the scalar interval geometry itself.

\subsection{Relative-Degree-Two HOCBFs and Input Compatibility}

For a continuously differentiable scalar function $\psi:\R^n\rightarrow\R$, define $L_A\psi(x):=\nabla\psi(x)^\trans Ax$ and $L_B\psi(x):=\nabla\psi(x)^\trans B$.
A continuous function $\alpha:\R\rightarrow\R$ belongs to the extended class $\cK_{\infty,e}$ if it is strictly increasing, satisfies $\alpha(0)=0$, and satisfies $\alpha(s)\rightarrow\pm\infty$ as $s\rightarrow\pm\infty$.

Let $\mathcal J$ be a finite index set, and let each $h_j\in\mathcal C^2(\R^n,\R)$ define $\cC_j:=\{x\in\R^n:h_j(x)\geq0\}$.
Each $h_j$ has input relative degree two on the domain of interest, namely $L_Bh_j(x)=0$ and $L_BL_Ah_j(x)\neq0$.

Let $\alpha_{j,1},\alpha_{j,2}\in\cK_{\infty,e}$, with $\alpha_{j,1}$ continuously differentiable and $\alpha_{j,2}$ locally Lipschitz.
Define
\begin{equation}
\phi_{j,1}(x):=L_Ah_j(x)+\alpha_{j,1}(h_j(x)).
\label{eq:hocbf-phi1-prelim}
\end{equation}
The corresponding HOCBF set is $\cC_{H,j}:=\{x\in\R^n:h_j(x)\geq0,\ \phi_{j,1}(x)\geq0\}$, and the joint set is $\cC_H:=\bigcap_{j\in\mathcal J}\cC_{H,j}$.

Following the HOCBF construction \cite{xiao2022,dasburdick2025}, the $j$th relative-degree-two safety condition is
\begin{equation}
L_A\phi_{j,1}(x)+L_B\phi_{j,1}(x)u+\alpha_{j,2}(\phi_{j,1}(x))\geq0.
\label{eq:hocbf-condition-prelim}
\end{equation}
Because $L_Bh_j=0$, differentiation of \eqref{eq:hocbf-phi1-prelim} with respect to the input direction gives $L_B\phi_{j,1}(x)=L_BL_Ah_j(x)$.
Hence, each relative-degree-two HOCBF condition is affine in the control input.

The simultaneous nominal HOCBF-admissible input set is
\begin{equation}
\begin{aligned}
\cU_H(x):=
\Big\{u\in\cU:\;&L_A\phi_{j,1}(x)+L_B\phi_{j,1}(x)u\\
&+\alpha_{j,2}(\phi_{j,1}(x))\geq0,\ \forall j\in\mathcal J\Big\}.
\end{aligned}
\label{eq:joint-hocbf-input-set-prelim}
\end{equation}
The multiple HOCBF constraints are pointwise compatible with the actuator set precisely when $\cU_H(x)\neq\varnothing$.
If a locally Lipschitz controller $u=k(x)$ satisfies $k(x)\in\cU_H(x)$ for every $x\in\cC_H$, then all conditions in \eqref{eq:hocbf-condition-prelim} hold simultaneously and $\cC_H$ is forward invariant \cite{xiao2022}.
Section~III constructs the observer-certified robust counterpart of \eqref{eq:joint-hocbf-input-set-prelim}, characterizes its polyhedral feasibility for general polytopic inputs, and derives the exact shared scalar-input geometry.

\section{Robust Safety Filtering}
\label{sec:proposed-framework}

\noindent\textbf{Problem statement.}
For the uncertain plant defined below, construct observer-certified robust HOCBF constraints and characterize their compatibility with hard polytopic actuator limits.
For the shared scalar-input case, obtain an exact feasibility characterization and minimum-modification filter, with robust forward invariance under maintained feasibility and finite-horizon deviation accounting relative to the baseline $H_\infty$ policy.

Figure~\ref{fig:proposed-framework} summarizes the resulting architecture.
The baseline $H_\infty$ controller provides the performance-oriented reference policy, while the disturbance observer supplies both the estimate used for nominal disturbance compensation and the certified estimation-error radius used in the robust HOCBF construction.
The safety filter intersects the resulting robust HOCBF requirements with the actuator set and modifies the nominal input only when required by this intersection.
Hereafter, $u$, $u_a$, and $u^\star$ denote a generic control input, the applied input, and the robust safety-filter solution, respectively.
\begin{figure*}[ht]
    \centering
    \includegraphics[width=1.5\columnwidth]{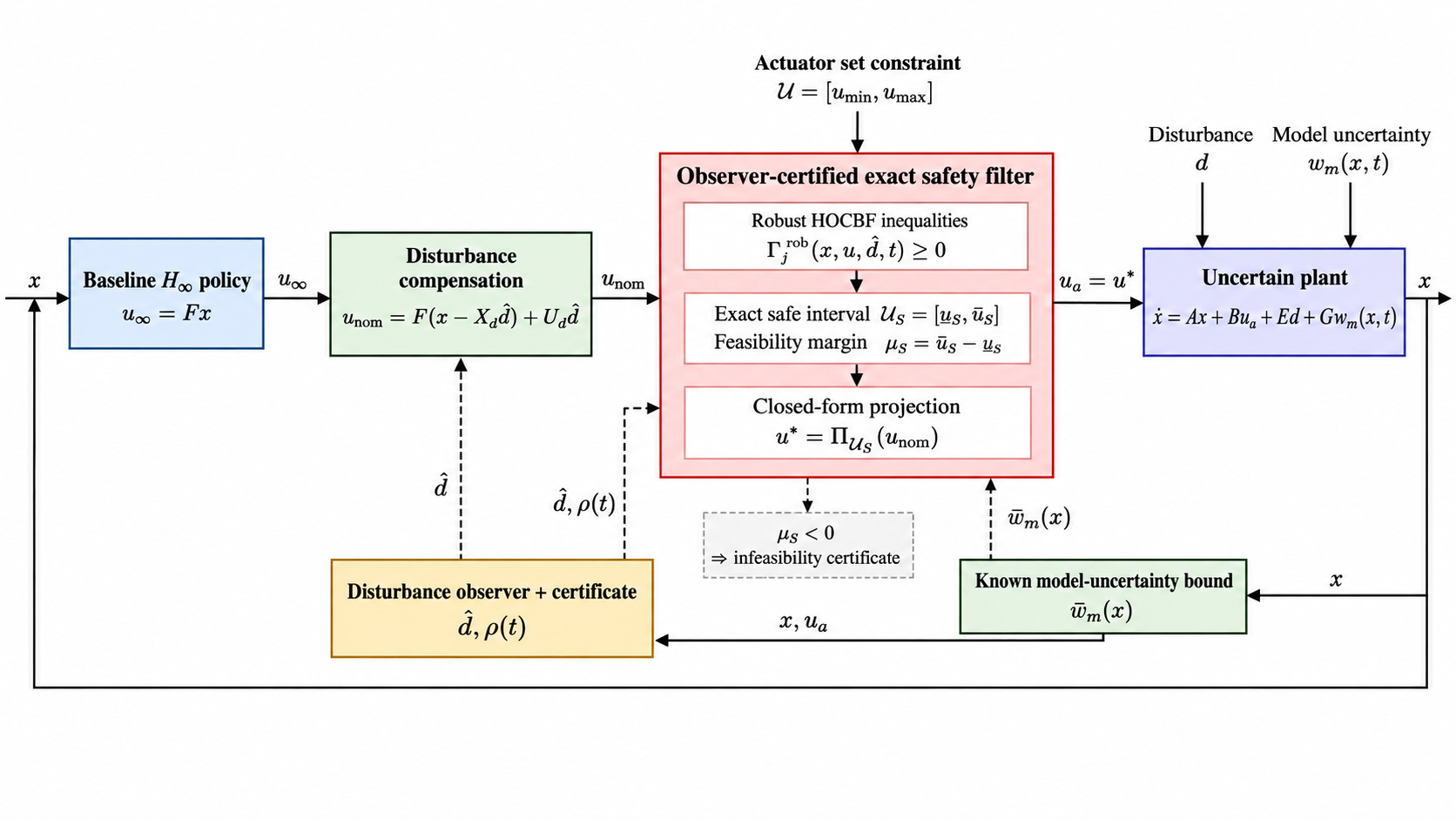}
    \caption{Architecture of the observer-certified robust safety filter. 
    The baseline $H_\infty$ policy and disturbance estimate define $u_{\mathrm{nom}}$, while the observer certificate $\rho(t)$ and model-uncertainty bound $\bar w_m(x)$ enter the robust HOCBF constraints.
    For the shared scalar input, intersection of these constraints with the actuator set yields the exact safe-input interval $\mathcal U_S=[\underline u_S,\overline u_S]$ and feasibility margin $\mu_S$.
    The applied input is the closed-form projection $u^\star=\Pi_{\mathcal U_S}(u_{\mathrm{nom}})$, while $\mu_S<0$ certifies that no admissible input satisfies all imposed robust HOCBF and actuator constraints at that point.}
    \label{fig:proposed-framework}
\end{figure*}
\subsection{Plant and Uncertainty Model}
\label{subsec:plant-uncertainty}

The nominal disturbed model used for the baseline $H_\infty$ synthesis is
\begin{equation}
\dot x=Ax+Bu+Ed,\qquad x(0)=x_0.
\label{eq:lti-system}
\end{equation}
Here, $x\in\R^n$, $u\in\R^m$, $d\in\R^p$, and $E\in\R^{n\times p}$ is a known disturbance distribution matrix.
The input $u$ in \eqref{eq:lti-system} is the unconstrained control variable used for the baseline $H_\infty$ synthesis.
The hard actuator set $\cU$ is imposed on the applied input through the safety filter.
The disturbance $d$ is unknown, and its temporal regularity is specified in Assumption~\ref{ass:disturbance}.

The actual plant used for the robustness and safety analysis additionally contains structured model uncertainty,
\begin{equation}
\dot x=Ax+Bu+Ed+Gw_m(x,t).
\label{eq:uncertain-lti-system}
\end{equation}
Here, $G\in\R^{n\times q}$ is known and specifies the uncertainty distribution channel, while $w_m(x,t)\in\R^q$ is unknown.
The disturbance $d$ and model uncertainty $w_m$ are treated as distinct uncertainty sources because the former is dynamically estimated, whereas the latter is retained through a pointwise certified bound.

\begin{assumption}[Structured model uncertainty]
\label{ass:model-uncertainty}
The signal $w_m(x,t)$ satisfies
\begin{equation}
\norm{w_m(x,t)}\leq\bar w_m(x),
\label{eq:model-uncertainty-bound}
\end{equation}
where $\bar w_m:\R^n\rightarrow\R_{\geq0}$ is known and locally Lipschitz.
For each fixed $t$, $w_m(\cdot,t)$ is locally Lipschitz in $x$, and for each fixed $x$, $w_m(x,\cdot)$ is piecewise continuous in $t$.
There exist known constants $\eta_0,\eta_x\geq0$ such that $\bar w_m(x)\leq\eta_0+\eta_x\norm{x}$.
\end{assumption}

Assumption~\ref{ass:model-uncertainty} covers structured additive uncertainty and unmodeled dynamics whose effect lies in the known subspace generated by $G$.
For example, a perturbation $\Delta_A(t)x+r(x,t)$ is included whenever $\Delta_A(t)x+r(x,t)=Gw_m(x,t)$ and the resulting $w_m$ satisfies \eqref{eq:model-uncertainty-bound}.
The distribution matrix $G$ is further restricted in the robust HOCBF construction so that this uncertainty does not alter the prescribed relative degree of the safety functions.
The formal guarantees do not cover uncertainty in the control matrix $B$, measurement noise, sampled-data effects, or nonlinear remainders introduced by model linearization unless their effects admit the structured representation $Gw_m$ in Assumption~\ref{ass:model-uncertainty}.

\subsection{Baseline $H_\infty$ Differential-Game Formulation}
\label{subsec:hinf-game}

The baseline performance policy is synthesized from the unconstrained differential game associated with the nominal disturbed model \eqref{eq:lti-system}.
The hard actuator constraint is enforced only on the applied input through the subsequent safety filter.
The structured uncertainty $w_m$ is retained explicitly in the observer and safety analysis rather than included in the baseline $H_\infty$ synthesis.
Let the $H_\infty$ performance output be $z=C_zx+D_zu$, where $C_z\in\R^{n_z\times n}$ and $D_z\in\R^{n_z\times m}$.
Define $Q:=C_z^\trans C_z\succeq0$ and $R:=D_z^\trans D_z\succ0$, and assume $C_z^\trans D_z=0$, so that $z^\trans z=x^\trans Qx+u^\trans Ru$.

For a prescribed attenuation level $\gamma>0$, consider the differential-game cost \cite{basar1995}
\begin{equation}
J_\gamma(u,d;x_0):=\int_0^\infty\left(x^\trans Qx+u^\trans Ru-\gamma^2d^\trans d\right)dt.
\label{eq:hinf-cost}
\end{equation}
Its upper value is $V_\gamma(x_0):=\inf_u\sup_d J_\gamma(u,d;x_0)$ over admissible causal strategies, with the minimizing control variable taking values in $\R^m$.

For the quadratic candidate $V(x)=x^\trans Px$, with $P=P^\trans\succeq0$, the corresponding Isaacs Hamiltonian is
\begin{equation}
\begin{aligned}
\cH_P(x,u,d):={}&x^\trans Qx+u^\trans Ru-\gamma^2d^\trans d\\
&+2x^\trans P(Ax+Bu+Ed).
\end{aligned}
\label{eq:isaacs-hamiltonian}
\end{equation}
Strict convexity in $u$ and strict concavity in $d$ give the saddle policies $u_\infty(x)=-R^{-1}B^\trans Px$ and $d_{\mathrm H}^{\star}(x)=\gamma^{-2}E^\trans Px$.
The associated generalized algebraic Riccati equation is
\begin{equation}
A^\trans P+PA+Q-PBR^{-1}B^\trans P+\gamma^{-2}PEE^\trans P=0.
\label{eq:gare}
\end{equation}
Define $F:=-R^{-1}B^\trans P$, so that $u_\infty=Fx$.
The baseline policy $u_\infty$ is used as an unconstrained performance reference and is not required to satisfy the hard actuator constraint pointwise.

\begin{assumption}[Stabilizing $H_\infty$ Riccati solution]
\label{ass:isaacs}
For the selected $\gamma>0$, \eqref{eq:gare} admits a symmetric solution $P\succeq0$ such that $A+BF$ and $A+BF+\gamma^{-2}EE^\trans P$ are Hurwitz.
The associated strategies satisfy the standard infinite-horizon admissibility and transversality conditions.
\end{assumption}

The following identity is a standard consequence of the Riccati equation and is retained as an auxiliary result because it provides the reference used later to quantify the effect of disturbance compensation and safety filtering.

\begin{proposition}[Baseline $H_\infty$ residual identity]
\label{prop:isaacs-residual}
Under Assumption~\ref{ass:isaacs}, the nominal model \eqref{eq:lti-system} satisfies, for every admissible input $u$ taking values in $\R^m$ and disturbance $d$,
\begin{equation}
\dot V+z^\trans z-\gamma^2d^\trans d=\norm{u-u_\infty}_R^2-\gamma^2\norm{d-d_{\mathrm H}^{\star}}^2.
\label{eq:isaacs-residual-identity}
\end{equation}
\end{proposition}

\begin{IEEEproof}
Using $u_\infty=-R^{-1}B^\trans Px$ gives $u^\trans Ru+2x^\trans PBu=\norm{u-u_\infty}_R^2-x^\trans PBR^{-1}B^\trans Px$.
Similarly, using $d_{\mathrm H}^{\star}=\gamma^{-2}E^\trans Px$ gives $-\gamma^2d^\trans d+2x^\trans PEd=-\gamma^2\norm{d-d_{\mathrm H}^{\star}}^2+\gamma^{-2}x^\trans PEE^\trans Px$.
Substituting these identities into $\dot V+z^\trans z-\gamma^2d^\trans d$ leaves the quadratic form associated with \eqref{eq:gare}, which is zero.
\end{IEEEproof}

For $u=u_\infty$, \eqref{eq:isaacs-residual-identity} gives $\dot V+z^\trans z-\gamma^2d^\trans d=-\gamma^2\norm{d-d_{\mathrm H}^{\star}}^2\leq0$ and therefore $\norm{z}_{\mathcal L_2}\leq\gamma\norm{d}_{\mathcal L_2}$ for $x_0=0$ and $d\in\mathcal L_2$.
Here, $u_\infty$ and $d_{\mathrm H}^{\star}$ are the nominal differential-game policies, with $d_{\mathrm H}^{\star}$ distinct from the online estimate of the realized disturbance.
For bounded-rate disturbances outside $\mathcal L_2[0,\infty)$, Proposition~\ref{prop:isaacs-residual} is used only through its finite-horizon identity.
\subsection{Disturbance Observer and Dynamic Error Certificate}
\label{subsec:disturbance-observer}

\begin{assumption}[Observer information structure]
\label{ass:observer-information}
The full state $x$ is available for feedback, and the applied input $u_a$ is available to the disturbance observer.
The disturbance distribution matrix $E$ has full column rank.
\end{assumption}

Under Assumption~\ref{ass:observer-information}, define the left pseudoinverse $E^\dagger:=(E^\trans E)^{-1}E^\trans$, so that $E^\dagger E=I_p$.
Because disturbance compensation and safety filtering may modify the baseline input, the observer is driven by the input $u_a$ actually applied to the plant.
Consider the disturbance observer
\begin{equation}
\begin{aligned}
\hat d&=z_d+\lambda_dE^\dagger x,\\
\dot z_d&=-\lambda_dE^\dagger\left(Ax+Bu_a+E\hat d\right),
\end{aligned}
\label{eq:disturbance-observer}
\end{equation}
where $\lambda_d>0$ is the observer gain.

Differentiating $\hat d$ and substituting the uncertain plant \eqref{eq:uncertain-lti-system} with $u=u_a$ gives
\begin{equation}
\dot{\hat d}=\lambda_d(d-\hat d)+\lambda_dE^\dagger Gw_m.
\label{eq:disturbance-estimate-dynamics}
\end{equation}
The terms $Ax$ and $Bu_a$ cancel because the observer uses the available state and the actually applied input.
Consequently, safety-filter intervention does not enter the disturbance-estimation error dynamics directly.
This cancellation is the reason the information structure in Assumption~\ref{ass:observer-information} is required by the disturbance certificate.
Defining $e_d:=d-\hat d$ gives
\begin{equation}
\dot e_d=\dot d-\lambda_de_d-\lambda_dE^\dagger Gw_m.
\label{eq:disturbance-error-uncertain}
\end{equation}
Thus, the estimation error is driven by two distinct sources: temporal variation of the disturbance through $\dot d$ and contamination from the structured model uncertainty through $E^\dagger Gw_m$.

\begin{assumption}[Disturbance regularity]
\label{ass:disturbance}
The disturbance $d$ is continuously differentiable and satisfies $\norm{\dot d(t)}\leq\omega_1$ for all $t\geq0$.
The constant $\omega_1\geq0$ is a prescribed disturbance-rate bound obtained from prior disturbance knowledge, physical slew-rate limits, or offline identification and does not require online measurement of $\dot d$.
The initial estimation error satisfies $\norm{e_d(0)}\leq\bar e_0$ for a known $\bar e_0\geq0$.
\end{assumption}

Define $\kappa_G:=\norm{E^\dagger G}$ and introduce the scalar certificate
\begin{equation}
\dot\rho=-\lambda_d\rho+\omega_1+\lambda_d\kappa_G\bar w_m(x),\qquad \rho(0)=\bar e_0.
\label{eq:dynamic-estimation-error-bound}
\end{equation}
Because the forcing terms in \eqref{eq:dynamic-estimation-error-bound} are nonnegative, $\rho(t)\geq0$ on its interval of existence.
If $\bar w_m(x(s))\leq\bar w_I$ on $[0,t]$, the variation-of-constants formula gives
\begin{equation}
\rho(t)\leq e^{-\lambda_dt}\bar e_0+\left(\frac{\omega_1}{\lambda_d}+\kappa_G\bar w_I\right)\left(1-e^{-\lambda_dt}\right).
\label{eq:dynamic-error-radius-uniform-bound}
\end{equation}
Thus, increasing $\lambda_d$ reduces the contribution $\omega_1/\lambda_d$ associated with disturbance variation but does not remove the contribution $\kappa_G\bar w_I$ induced by structured model uncertainty.

\begin{lemma}[Observer-certified estimation-error bound]
\label{lem:estimation-error-bound}
Under Assumptions~\ref{ass:model-uncertainty}, \ref{ass:observer-information}, and~\ref{ass:disturbance}, along every Carathéodory solution of \eqref{eq:uncertain-lti-system}, \eqref{eq:disturbance-observer}, and \eqref{eq:dynamic-estimation-error-bound}, the disturbance estimation error satisfies $\norm{e_d(t)}\leq\rho(t)$ for every $t$ in the interval of existence of the solution.
\end{lemma}

\begin{IEEEproof}
For $e_d\neq0$, differentiating its Euclidean norm along \eqref{eq:disturbance-error-uncertain} gives
\begin{equation}
\begin{aligned}
\frac{d}{dt}\norm{e_d}
&=\frac{e_d^\trans\dot e_d}{\norm{e_d}}\\
&\leq-\lambda_d\norm{e_d}+\norm{\dot d}+\lambda_d\norm{E^\dagger G}\norm{w_m}.
\end{aligned}
\label{eq:observer-error-norm-bound}
\end{equation}
At $e_d=0$, the same inequality holds in the upper-right Dini-derivative sense.
Applying Assumptions~\ref{ass:model-uncertainty} and~\ref{ass:disturbance} to \eqref{eq:observer-error-norm-bound} yields
\begin{equation}
D^+\norm{e_d}\leq-\lambda_d\norm{e_d}+\omega_1+\lambda_d\kappa_G\bar w_m(x).
\label{eq:observer-error-dini-bound}
\end{equation}
The right-hand side of \eqref{eq:observer-error-dini-bound} is the same scalar comparison dynamics used to define $\rho$ in \eqref{eq:dynamic-estimation-error-bound}.
Because $\rho(0)=\bar e_0\geq\norm{e_d(0)}$, the comparison lemma gives $\norm{e_d(t)}\leq\rho(t)$ throughout the common interval of existence.
\end{IEEEproof}

Lemma~\ref{lem:estimation-error-bound} converts the unknown disturbance into a pointwise uncertainty set parameterized by the online estimate and its certified error radius.
Define
\begin{equation}
\mathcal D(\hat d,\rho):=\left\{\delta\in\R^p:\norm{\delta-\hat d}\leq\rho\right\}.
\label{eq:safety-adversary-set}
\end{equation}
Along the closed-loop trajectory, $d(t)\in\mathcal D(\hat d(t),\rho(t))$ throughout the interval of existence.
The structured model uncertainty is represented separately by
\begin{equation}
\mathcal W(x):=\left\{w\in\R^q:\norm{w}\leq\bar w_m(x)\right\},
\label{eq:model-uncertainty-set}
\end{equation}
for which Assumption~\ref{ass:model-uncertainty} gives $w_m(x,t)\in\mathcal W(x)$.

The set $\mathcal D(\hat d,\rho)$ is an observer-certified set for the realized disturbance and is distinct from the maximizing policy $d_{\mathrm H}^{\star}(x)$ of the nominal $H_\infty$ differential game.
The Cartesian product $\mathcal D(\hat d,\rho)\times\mathcal W(x)$ is used as the certified pointwise uncertainty envelope in the robust HOCBF analysis.
Because $e_d$ and $w_m$ are dynamically coupled through \eqref{eq:disturbance-error-uncertain}, this Cartesian product is not claimed to equal the true joint reachable uncertainty set.
It is a certified outer description from which the robust safety inequalities can be evaluated directly.

\subsection{Disturbance-Compensated Nominal Control}
\label{subsec:disturbance-compensation}

The baseline policy $u_\infty=Fx$ does not use the online estimate $\hat d$.
To introduce disturbance compensation, define the regulated output $y_r:=C_rx$ and suppose there exist matrices $X_d\in\R^{n\times p}$ and $U_d\in\R^{m\times p}$ satisfying
\begin{equation}
AX_d+BU_d+E=0,\qquad C_rX_d=0.
\label{eq:regulator-equations}
\end{equation}
For a constant disturbance, \eqref{eq:regulator-equations} defines the zero-offset equilibrium $(x,u)=(X_dd,U_dd)$ of the nominal disturbed model.
Replacing $d$ by its online estimate gives
\begin{equation}
u_{\mathrm{nom}}(x,\hat d):=F(x-X_d\hat d)+U_d\hat d= u_\infty+K_{\hat d}\hat d,
\label{eq:disturbance-compensated-nominal-input}
\end{equation}
where $K_{\hat d}:=U_d-FX_d$.
Thus, $u_\infty$ is the differential-game baseline, $u_{\mathrm{nom}}$ is its disturbance-compensated counterpart, and $u^\star$ is the subsequent safety-filtered input.
\subsection{Observer-Certified Robust HOCBF Constraints}
\label{subsec:robust-hocbf}

For the safety functions introduced in Section~\ref{sec:preliminaries}, define $L_E\psi(x):=\nabla\psi(x)^\trans E$ and $L_G\psi(x):=\nabla\psi(x)^\trans G$.
We retain the general extended class-$\cK_{\infty,e}$ functions $\alpha_{j,1}$ and $\alpha_{j,2}$ introduced in Section~\ref{sec:preliminaries}.
In addition to the relative-degree-two condition with respect to the control input, assume that neither uncertainty source enters the first derivative of the prescribed safety functions,
\begin{equation}
L_Bh_j=L_Eh_j=L_Gh_j=0,\qquad L_BL_Ah_j\neq0,\qquad j\in\mathcal J.
\label{eq:relative-degree-two-uncertainty}
\end{equation}
Define
\begin{equation}\beta_j(x):=L_A\phi_{j,1}(x)+\alpha_{j,2}\!\left(\phi_{j,1}(x)\right).
\label{eq:general-hocbf-beta}
\end{equation}
The coefficients multiplying the control input, disturbance, and structured model uncertainty are
\begin{equation}
\begin{aligned}
a_j(x)&:=L_BL_Ah_j(x)\in\R^{1\times m},\\
c_j(x)&:=L_EL_Ah_j(x)\in\R^{1\times p},\\
g_j(x)&:=L_GL_Ah_j(x)\in\R^{1\times q}.
\end{aligned}
\label{eq:hocbf-uncertainty-coefficients}
\end{equation}
Because $L_Bh_j=L_Eh_j=L_Gh_j=0$, the terms generated by differentiating $\alpha_{j,1}(h_j)$ vanish along the control and uncertainty directions.
Condition~\eqref{eq:relative-degree-two-uncertainty} therefore gives $L_B\phi_{j,1}=a_j$, $L_E\phi_{j,1}=c_j$, and $L_G\phi_{j,1}=g_j$.
Hence, along the uncertain plant \eqref{eq:uncertain-lti-system},
\begin{equation}
\dot\phi_{j,1}
+\alpha_{j,2}\!\left(\phi_{j,1}\right)=
\beta_j(x)+a_j(x)u+c_j(x)d+g_j(x)w_m.
\label{eq:uncertain-hocbf-dynamics}
\end{equation}
The $j$th uncertain HOCBF requirement is therefore
\begin{equation}
\beta_j(x)+a_j(x)u+c_j(x)d+g_j(x)w_m\geq0.
\label{eq:uncertain-hocbf-condition}
\end{equation}

Lemma~\ref{lem:estimation-error-bound} gives $d(t)\in\mathcal D(\hat d(t),\rho(t))$, while Assumption~\ref{ass:model-uncertainty} gives $w_m(x,t)\in\mathcal W(x)$.
For fixed $(x,\hat d,\rho,u)$, define the worst-case HOCBF value over the certified product envelope by
\begin{equation}
\begin{aligned}
\underline\Gamma_j(x,u,\hat d,\rho)
:=\min_{\substack{\delta\in\mathcal D(\hat d,\rho)\\w\in\mathcal W(x)}}\;&
\beta_j(x)+a_j(x)u\\
&+c_j(x)\delta+g_j(x)w.
\end{aligned}
\label{eq:robust-hocbf-minimization}
\end{equation}

\begin{proposition}[Exact robust counterpart]
\label{prop:exact-robust-hocbf}
For every $j\in\mathcal J$, the minimization in \eqref{eq:robust-hocbf-minimization} admits the closed-form value
\begin{equation}
\begin{aligned}
\underline\Gamma_j
={}&\beta_j(x)+a_j(x)u+c_j(x)\hat d\\
&-\rho\norm{c_j(x)}
-\bar w_m(x)\norm{g_j(x)}.
\end{aligned}
\label{eq:robust-hocbf-worst-case-value}
\end{equation}
Consequently, satisfaction of \eqref{eq:uncertain-hocbf-condition} for every $(\delta,w)\in\mathcal D(\hat d,\rho)\times\mathcal W(x)$ is equivalent to
\begin{equation}
\Gamma_j^{\mathrm{rob}}(x,u,\hat d,\rho):=\underline\Gamma_j(x,u,\hat d,\rho)\geq0.
\label{eq:dynamic-robust-hocbf}
\end{equation}
\end{proposition}

\begin{IEEEproof}
Because $\mathcal D(\hat d,\rho)\times\mathcal W(x)$ is a Cartesian product and the uncertain terms are separable,
\begin{equation}
\begin{aligned}
\underline\Gamma_j
={}&\beta_j+a_ju+\min_{\delta\in\mathcal D(\hat d,\rho)}c_j\delta\\
&+\min_{w\in\mathcal W(x)}g_jw.
\end{aligned}
\label{eq:robust-hocbf-separable-min}
\end{equation}
Writing $\delta=\hat d+e$ with $\norm{e}\leq\rho$ and applying Cauchy--Schwarz gives $c_je\geq-\rho\norm{c_j}$.
Equality is attained for $e=-\rho c_j^\trans/\norm{c_j}$ when $c_j\neq0$, while the result is immediate when $c_j=0$.
Hence, $\min_{\delta\in\mathcal D(\hat d,\rho)}c_j\delta=c_j\hat d-\rho\norm{c_j}$.
Likewise, $\min_{w\in\mathcal W(x)}g_jw=-\bar w_m(x)\norm{g_j}$, with equality attained in the direction opposite to $g_j^\trans$ whenever $g_j\neq0$.
Substitution into \eqref{eq:robust-hocbf-separable-min} gives \eqref{eq:robust-hocbf-worst-case-value}.
The equivalence in \eqref{eq:dynamic-robust-hocbf} follows because $\underline\Gamma_j$ is the minimum of the uncertain HOCBF expression over the certified product envelope.
\end{IEEEproof}

Proposition~\ref{prop:exact-robust-hocbf} is exact over the certified product envelope $\mathcal D(\hat d,\rho)\times\mathcal W(x)$.
Because $w_m$ also enters the observer-error dynamics through \eqref{eq:disturbance-error-uncertain}, this product set can remain an outer approximation of the true dynamically coupled uncertainty set.

Define the robust safety-admissible input set by
\begin{equation}
\mathcal U_S(x,\hat d,\rho):=\left\{u\in\mathcal U:
\Gamma_j^{\mathrm{rob}}(x,u,\hat d,\rho)\geq0,\ \forall j\in\mathcal J \right\}.
\label{eq:robust-safe-input-set}
\end{equation}
Thus, $\mathcal U_S$ is precisely the intersection of the actuator set with all observer-certified robust HOCBF input sets.

\begin{proposition}[Robust forward invariance]
\label{prop:robust-forward-invariance}
Suppose Assumptions~\ref{ass:model-uncertainty}, \ref{ass:observer-information}, and~\ref{ass:disturbance} hold and $\mathcal U$ is compact.
Suppose $h_j(x(0))\geq0$ and $\phi_{j,1}(x(0))\geq0$ for every $j\in\mathcal J$.
If $\mathcal U_S(x(t),\hat d(t),\rho(t))$ remains nonempty along the closed-loop trajectory and $u_a$ admits a locally Lipschitz feedback realization in a neighborhood of that trajectory with $u_a(x,\hat d,\rho)\in\mathcal U_S(x,\hat d,\rho)$ at every feasible point, then the extended closed loop admits a unique forward-complete Carathéodory solution and $\cC_H$ is forward invariant.
\end{proposition}

\begin{IEEEproof}
Assumption~\ref{ass:disturbance} gives $\norm{d(t)}\leq\norm{d(0)}+\omega_1t$, so $d$ is bounded on every finite time interval.
Because $u_a\in\mathcal U$ and $\mathcal U$ is compact, the applied input is bounded.
Assumption~\ref{ass:model-uncertainty} gives $\norm{w_m(x,t)}\leq\eta_0+\eta_x\norm{x}$, so the plant dynamics have at most linear growth in $x$ on every finite time interval.
Equation~\eqref{eq:disturbance-estimate-dynamics} is affine in $\hat d$ with forcing terms depending on $d$ and $w_m$, while \eqref{eq:dynamic-estimation-error-bound} is affine in $\rho$ with forcing of at most linear growth in $x$.
Hence, the extended dynamics in $(x,\hat d,\rho)$ have at most linear growth on every finite time interval.
Under the stated local Lipschitz and piecewise-continuity assumptions, the extended vector field satisfies the standard Carathéodory existence and local uniqueness conditions.
The resulting solution therefore exists uniquely on its maximal interval of existence.
The linear-growth bound excludes finite escape, so this maximal interval is $[0,\infty)$.

Fix any $j\in\mathcal J$.
Because $u_a\in\mathcal U_S$, Proposition~\ref{prop:exact-robust-hocbf} gives $\beta_j+a_ju_a+c_jd+g_jw_m\geq\Gamma_j^{\mathrm{rob}}(x,u_a,\hat d,\rho)\geq0$.
Equation~\eqref{eq:uncertain-hocbf-dynamics} therefore gives $\dot\phi_{j,1}+\alpha_{j,2}(\phi_{j,1})\geq0$.
Since $\alpha_{j,2}\in\cK_{\infty,e}$ and $\phi_{j,1}(x(0))\geq0$, scalar comparison gives $\phi_{j,1}(x(t))\geq0$.
Thus, $\dot h_j+\alpha_{j,1}(h_j)=\phi_{j,1}\geq0$, and the same argument with $h_j(x(0))\geq0$ gives $h_j(x(t))\geq0$.
Since $j$ is arbitrary, $\cC_H$ is forward invariant.
\end{IEEEproof}

\subsection{Robust Safety Filter and Input Feasibility}
\label{subsec:robust-safety-filter}

Whenever $\mathcal U_S(x,\hat d,\rho)$ is nonempty, the safety filter enforces the hard actuator and robust HOCBF constraints by selecting the feasible input closest to the disturbance-compensated nominal input according to
\begin{equation}
u^\star(x,\hat d,\rho):=
\arg\min_{u\in\mathcal U_S(x,\hat d,\rho)}
\frac{1}{2}\norm{u-u_{\mathrm{nom}}}_R^2.
\label{eq:robust-safety-filter-qp}
\end{equation}
The applied input is set to $u_a=u^\star$.
The weighting matrix $R\succ0$ is inherited from the $H_\infty$ differential-game cost and therefore measures intervention in the same control metric as the baseline design.
Forward invariance, however, depends on membership in $\mathcal U_S$ rather than on the particular positive-definite weighting used in \eqref{eq:robust-safety-filter-qp}.

Rearranging \eqref{eq:dynamic-robust-hocbf} gives $a_j(x)u\geq b_j(x,\hat d,\rho)$, where $b_j:=-\beta_j-c_j\hat d+\rho\norm{c_j}+\bar w_m\norm{g_j}$.

\begin{proposition}[General-input polyhedral feasibility]
\label{prop:general-input-feasibility}
Let $H_u u\leq h_u$ be an H-representation of the compact actuator polytope $\mathcal U$.
Define $A_S(x):=\operatorname{col}_{j\in\mathcal J}a_j(x)$ and $b_S(x,\hat d,\rho):=\operatorname{col}_{j\in\mathcal J}b_j(x,\hat d,\rho)$.
Then the robust safety-admissible input set is the polytope
\begin{equation}
\mathcal U_S(x,\hat d,\rho)=\left\{u\in\R^m:H_uu\leq h_u,\ A_S(x)u\geq b_S(x,\hat d,\rho)\right\}.
\label{eq:general-robust-safe-polytope}
\end{equation}
Hence, pointwise robust feasibility for any $m\geq1$ is equivalent to feasibility of the finite system of linear inequalities in \eqref{eq:general-robust-safe-polytope}.
Let $\bar H:=\operatorname{col}(H_u,-A_S)$ and $\bar h:=\operatorname{col}(h_u,-b_S)$.
The set $\mathcal U_S$ is empty if and only if there exists $y\geq0$ such that $y^\trans\bar H=0$ and $y^\trans\bar h<0$.
\end{proposition}

\begin{IEEEproof}
Intersecting $A_Su\geq b_S$ with $H_uu\leq h_u$ gives \eqref{eq:general-robust-safe-polytope}.
The feasibility characterization is immediate, while the infeasibility condition follows from Farkas' lemma applied to $\bar Hu\leq\bar h$.
\end{IEEEproof}

For $m>1$, Proposition~\ref{prop:general-input-feasibility} gives an exact polyhedral feasibility characterization, but the resulting safe-input geometry generally does not reduce to two scalar endpoints.
We now specialize to the shared scalar-input case $m=1$ with $\mathcal U=[u_{\min},u_{\max}]$, for which the polytope in \eqref{eq:general-robust-safe-polytope} reduces to an interval.
The coefficient $a_j(x)\in\R^{1\times1}$ is identified with its scalar entry.
The general robust inequality becomes
\begin{equation}
a_j(x)u\geq b_j(x,\hat d,\rho).
\label{eq:scalar-robust-hocbf-inequality}
\end{equation}
Condition~\eqref{eq:relative-degree-two-uncertainty} implies $a_j(x)\neq0$ for every $j\in\mathcal J$.

Define the state-dependent index sets $\mathcal J_+(x):=\{j\in\mathcal J:a_j(x)>0\}$ and $\mathcal J_-(x):=\{j\in\mathcal J:a_j(x)<0\}$.
For $j\in\mathcal J_+(x)$, \eqref{eq:scalar-robust-hocbf-inequality} is equivalent to $u\geq b_j/a_j$.
For $j\in\mathcal J_-(x)$, it is equivalent to $u\leq b_j/a_j$.
Thus, the general polyhedral characterization in Proposition~\ref{prop:general-input-feasibility} reduces exactly to the half-line intersection
\begin{equation}
\begin{aligned}
\mathcal U_S
={}&[u_{\min},u_{\max}]\\
&\cap\bigcap_{j\in\mathcal J_+(x)}
\left[\frac{b_j}{a_j},\infty\right)
\cap\bigcap_{j\in\mathcal J_-(x)}
\left(-\infty,\frac{b_j}{a_j}\right].
\end{aligned}
\label{eq:scalar-safe-set-intersection}
\end{equation}

Using $\max\varnothing=-\infty$ and $\min\varnothing=+\infty$, define
\begin{equation}
\begin{aligned}
\underline u_S(x,\hat d,\rho)
&:=\max\left\{u_{\min},\max_{j\in\mathcal J_+(x)}\frac{b_j}{a_j}\right\},\\
\overline u_S(x,\hat d,\rho)
&:=\min\left\{u_{\max},\min_{j\in\mathcal J_-(x)}\frac{b_j}{a_j}\right\}.
\end{aligned}
\label{eq:scalar-safe-interval-bounds}
\end{equation}
The first quantity is the largest lower bound imposed jointly by the actuator and robust HOCBF constraints, while the second is the smallest corresponding upper bound.
Therefore,
\begin{equation}
\mathcal U_S=
\begin{cases}
[\underline u_S,\overline u_S], & \underline u_S\leq\overline u_S,\\
\varnothing, & \underline u_S>\overline u_S.
\end{cases}
\label{eq:scalar-safe-input-set}
\end{equation}

Define the signed feasibility margin as $\mu_S(x,\hat d,\rho):=\overline u_S(x,\hat d,\rho)-\underline u_S(x,\hat d,\rho)$.
Its sign determines feasibility, while its magnitude measures the remaining shared control authority.

\begin{proposition}[Exact shared-input feasibility]
\label{prop:exact-scalar-feasibility}
For each fixed $(x,\hat d,\rho)$, consider the scalar specialization of the polyhedral robust safety-admissible set in Proposition~\ref{prop:general-input-feasibility}.
The simultaneous robust HOCBF and actuator constraints admit a common input if and only if
\begin{equation}
\underline u_S(x,\hat d,\rho)\leq\overline u_S(x,\hat d,\rho).
\label{eq:exact-scalar-feasibility-condition}
\end{equation}
Whenever \eqref{eq:exact-scalar-feasibility-condition} holds, the feasible set is exactly $\mathcal U_S=[\underline u_S,\overline u_S]$ and the unique solution of \eqref{eq:robust-safety-filter-qp} is
\begin{equation}
u^\star=
\min\left\{
\overline u_S,
\max\left\{\underline u_S,u_{\mathrm{nom}}\right\}
\right\}.
\label{eq:scalar-safety-projection}
\end{equation}
Moreover, $\mu_S>0$, $\mu_S=0$, and $\mu_S<0$ correspond, respectively, to a feasible interval with nonempty interior, a singleton feasible set, and an empty robust safety-admissible set.
\end{proposition}

\begin{IEEEproof}
For $j\in\mathcal J_+(x)$ and $j\in\mathcal J_-(x)$, $a_ju\geq b_j$ defines the lower and upper half-lines $[b_j/a_j,\infty)$ and $(-\infty,b_j/a_j]$, respectively.
Their intersection with $\mathcal U$ gives \eqref{eq:scalar-safe-set-intersection}, which is nonempty exactly under \eqref{eq:exact-scalar-feasibility-condition}.
When feasible, $\mathcal U_S=[\underline u_S,\overline u_S]$, and strict convexity for $R>0$ makes \eqref{eq:robust-safety-filter-qp} the projection \eqref{eq:scalar-safety-projection}.
The three cases for $\mu_S$ follow from the signed separation of the endpoints.
\end{IEEEproof}

Unlike the general polyhedral characterization in Proposition~\ref{prop:general-input-feasibility}, the scalar specialization provides explicit endpoints, a signed feasibility reserve, and the closed-form projection in \eqref{eq:scalar-safety-projection}.

\begin{corollary}[Domain-wise feasibility certificate]
\label{cor:domain-feasibility}
Let $\Omega\subset\R^n\times\R^p\times\R_{\geq0}$ be a compact connected set on which $a_j(x)$ and $b_j(x,\hat d,\rho)$ are continuous and $a_j(x)\neq0$ for every $j\in\mathcal J$.
Let $\xi:=(x,\hat d,\rho)$, and use the shorthand $a_j(\xi):=a_j(x)$ and $b_j(\xi):=b_j(x,\hat d,\rho)$.
Because $a_j$ is continuous and nonzero on the connected set $\Omega$, its sign is constant on $\Omega$.
Define $\mathcal J_+^\Omega:=\{j\in\mathcal J:a_j(\xi)>0,\ \forall\xi\in\Omega\}$ and $\mathcal J_-^\Omega:=\{j\in\mathcal J:a_j(\xi)<0,\ \forall\xi\in\Omega\}$.
Using the same empty-set conventions as in \eqref{eq:scalar-safe-interval-bounds}, define
\begin{equation}
\begin{aligned}
\ell_\Omega&:=\max_{j\in\mathcal J_+^\Omega}\sup_{\xi\in\Omega}\frac{b_j(\xi)}{a_j(\xi)},\\
r_\Omega&:=\min_{j\in\mathcal J_-^\Omega}\inf_{\xi\in\Omega}\frac{b_j(\xi)}{a_j(\xi)}.
\end{aligned}
\label{eq:domain-hocbf-demand-bounds}
\end{equation}
Define $\underline u_\Omega:=\max\{u_{\min},\ell_\Omega\}$, $\overline u_\Omega:=\min\{u_{\max},r_\Omega\}$, and $\mu_\Omega:=\overline u_\Omega-\underline u_\Omega$.
If $\mu_\Omega\geq0$, then $\mu_S(x,\hat d,\rho)\geq\mu_\Omega\geq0$ for every $(x,\hat d,\rho)\in\Omega$.
Hence, $\mathcal U_S(x,\hat d,\rho)$ is nonempty throughout $\Omega$.
\end{corollary}

\begin{IEEEproof}
For every $\xi\in\Omega$, the definition of $\ell_\Omega$ gives $\underline u_S(\xi)\leq\underline u_\Omega$.
Similarly, the definition of $r_\Omega$ gives $\overline u_S(\xi)\geq\overline u_\Omega$.
Therefore, $\mu_S(\xi)=\overline u_S(\xi)-\underline u_S(\xi)\geq\overline u_\Omega-\underline u_\Omega=\mu_\Omega$.
The condition $\mu_\Omega\geq0$ then implies pointwise feasibility throughout $\Omega$ by Proposition~\ref{prop:exact-scalar-feasibility}.
\end{IEEEproof}

\begin{remark}[Actuator-authority interpretation]
\label{rem:actuator-authority}
The bounds in \eqref{eq:domain-hocbf-demand-bounds} separate the robust HOCBF control demand from the available actuator authority.
The quantity $\ell_\Omega$ is the largest lower control demand on $\Omega$, while $r_\Omega$ is the smallest upper control demand.
Hence, a sufficient domain-wise actuator-sizing condition is
\begin{equation}
\max\{u_{\min},\ell_\Omega\}\leq\min\{u_{\max},r_\Omega\}.
\label{eq:domain-actuator-sizing-condition}
\end{equation}
For a symmetric actuator set $\mathcal U=[-\bar u,\bar u]$, \eqref{eq:domain-actuator-sizing-condition} is satisfied whenever $\ell_\Omega\leq r_\Omega$ and $\bar u\geq\max\{\ell_\Omega,-r_\Omega\}$.
Thus, $\ell_\Omega\leq r_\Omega$ characterizes compatibility of the robust HOCBF control demands on the prescribed domain, while the actuator bound determines whether their common intersection can be realized.
\end{remark}

The domain certificate can also be related directly to the observer and uncertainty parameters.
For a prescribed compact state domain $\mathcal X_0$, define $\bar w_0:=\sup_{x\in\mathcal X_0}\bar w_m(x)$.
If $x(t)\in\mathcal X_0$, then \eqref{eq:dynamic-error-radius-uniform-bound} gives $\rho(t)\leq\bar\rho_0:=\max\{\bar e_0,\omega_1/\lambda_d+\kappa_G\bar w_0\}$.
For any prescribed compact estimate domain $\widehat{\mathcal D}_0\subset\R^p$, Corollary~\ref{cor:domain-feasibility} can therefore be evaluated over $\Omega_0:=\mathcal X_0\times\widehat{\mathcal D}_0\times[0,\bar\rho_0]$.
The resulting actuator-sizing condition depends explicitly on $\omega_1$, $\lambda_d$, the structured model-uncertainty bound, and the actuator limits.
This certificate guarantees feasibility on $\Omega_0$ but does not establish that the closed-loop trajectory remains inside $\Omega_0$.

\begin{corollary}[Local Lipschitz regularity]
\label{cor:lipschitz-filter}
Suppose $a_j(x)$ and $u_{\mathrm{nom}}(x,\hat d)$ are locally Lipschitz, and $b_j(x,\hat d,\rho)$ is locally Lipschitz in $(x,\hat d,\rho)$.
Suppose also that $a_j(x)\neq0$ for every $j\in\mathcal J$ on the domain of interest.
Then $\underline u_S$ and $\overline u_S$ are locally Lipschitz in $(x,\hat d,\rho)$.
Moreover, the projection law in \eqref{eq:scalar-safety-projection} admits a locally Lipschitz extension to a neighborhood of every point satisfying $\mu_S\geq0$.
On the feasible set $\mu_S\geq0$, this extension coincides with the unique safety-filter solution $u^\star$.
\end{corollary}

\begin{IEEEproof}
Since $a_j$ is locally Lipschitz and nonzero, its sign is locally constant and $1/a_j$ is locally Lipschitz.
Thus, every active ratio $b_j/a_j$ is locally Lipschitz.
Finite maxima and minima preserve local Lipschitz continuity, so $\underline u_S$ and $\overline u_S$ are locally Lipschitz.
The algebraic map $\widetilde u:=\min\{\overline u_S,\max\{\underline u_S,u_{\mathrm{nom}}\}\}$ is locally Lipschitz as a finite composition of locally Lipschitz minimum and maximum operations.
Whenever $\mu_S\geq0$, Proposition~\ref{prop:exact-scalar-feasibility} gives $\mathcal U_S=[\underline u_S,\overline u_S]$, and therefore $\widetilde u=u^\star$.
Hence, $\widetilde u$ provides a locally Lipschitz extension of the safety-filter solution through feasible boundary points satisfying $\mu_S=0$.
\end{IEEEproof}

\subsection{Finite-Horizon Deviation from the $H_\infty$ Baseline}
\label{subsec:performance-accounting}

The baseline residual identity in Proposition~\ref{prop:isaacs-residual} is derived for the nominal disturbed model.
Along the uncertain plant \eqref{eq:uncertain-lti-system}, the additional term $Gw_m$ contributes $2x^\trans PGw_m$ to $\dot V$.
Therefore,
\begin{equation}
\begin{aligned}
\dot V+z^\trans z-\gamma^2d^\trans d
={}&\norm{u-u_\infty}_R^2-\gamma^2\norm{d-d_{\mathrm H}^{\star}}^2\\
&+2x^\trans PGw_m.
\end{aligned}
\label{eq:uncertain-hinf-residual}
\end{equation}

Since $u^\star-u_\infty=(u_{\mathrm{nom}}-u_\infty)+(u^\star-u_{\mathrm{nom}})$, deviation from the differential-game baseline contains both disturbance compensation and safety-filter intervention.
Accordingly, $J_{\mathrm{dev}}$ below measures their combined effect and should not be interpreted as safety-filter intervention alone.

Define the accumulated baseline deviation as $J_{\mathrm{dev}}(T):=\int_0^T\norm{u^\star-u_\infty}_R^2dt$.
Define also the exact model-uncertainty contribution as $J_m(T):=2\int_0^Tx^\trans PGw_mdt$ and its certified trajectory-dependent bound as $\overline J_m(T):=2\int_0^T\norm{G^\trans Px}\bar w_m(x)dt$.

\begin{proposition}[Finite-horizon deviation balance]
\label{prop:finite-horizon-balance}
Suppose Assumptions~\ref{ass:isaacs} and~\ref{ass:model-uncertainty} hold and $\mathcal U_S(x(t),\hat d(t),\rho(t))$ remains nonempty for all $t\in[0,T]$.
Let $u^\star$ be the corresponding safety-filter input and define $z^\star:=C_zx+D_zu^\star$.
Then
\begin{equation}
\begin{aligned}
V(x(T))
&+\int_0^T(z^\star)^\trans z^\star dt\\
&+\gamma^2\int_0^T\norm{d-d_{\mathrm H}^{\star}}^2dt\\
&=V(x_0)+\gamma^2\int_0^T\norm{d}^2dt\\
&\quad+J_{\mathrm{dev}}(T)+J_m(T).
\end{aligned}
\label{eq:finite-horizon-baseline-deviation-balance}
\end{equation}
Moreover, the model-uncertainty contribution satisfies
\begin{equation}
\abs{J_m(T)}\leq\overline J_m(T),
\label{eq:model-uncertainty-energy-bound}
\end{equation}
and consequently
\begin{equation}
\begin{aligned}
V(x(T))
&+\int_0^T(z^\star)^\trans z^\star dt\\
&+\gamma^2\int_0^T\norm{d-d_{\mathrm H}^{\star}}^2dt\\
&\leq V(x_0)+\gamma^2\int_0^T\norm{d}^2dt\\
&\quad+J_{\mathrm{dev}}(T)+\overline J_m(T).
\end{aligned}
\label{eq:finite-horizon-baseline-deviation-bound}
\end{equation}
\end{proposition}

\begin{IEEEproof}
Setting $u=u^\star$ and $z=z^\star$ in \eqref{eq:uncertain-hinf-residual}, integrating over $[0,T]$, and rearranging gives \eqref{eq:finite-horizon-baseline-deviation-balance}.
By Cauchy--Schwarz and Assumption~\ref{ass:model-uncertainty},
\begin{equation}
\begin{aligned}
\abs{J_m(T)}&\leq2\int_0^T\norm{G^\trans Px}\norm{w_m}\,dt\\
&\leq2\int_0^T\norm{G^\trans Px}\bar w_m(x)\,dt=\overline J_m(T).
\end{aligned}
\end{equation}
Substitution into \eqref{eq:finite-horizon-baseline-deviation-balance} gives \eqref{eq:finite-horizon-baseline-deviation-bound}.
\end{IEEEproof}

For $x_0=0$, $V(x(T))\geq0$ and the disturbance-mismatch integral in \eqref{eq:finite-horizon-baseline-deviation-bound} is nonnegative.
Dropping these terms gives
\begin{equation}
\int_0^T(z^\star)^\trans z^\star dt
\leq
\gamma^2\int_0^T\norm{d}^2dt
+J_{\mathrm{dev}}(T)
+\overline J_m(T).
\label{eq:filtered-output-energy-bound}
\end{equation}
The term $J_{\mathrm{dev}}(T)$ measures the total finite-horizon departure of the applied input from the baseline $H_\infty$ policy and therefore includes both disturbance compensation and safety intervention.
The term $\overline J_m(T)$ bounds the additional contribution associated with the structured model uncertainty.
Both terms depend on the realized closed-loop trajectory and are used here for performance accounting rather than as online constraints.

Equation~\eqref{eq:filtered-output-energy-bound} is therefore a finite-horizon deviation characterization rather than an $H_\infty$ attenuation guarantee for the compensated and safety-filtered closed loop.
The nominal finite-horizon $H_\infty$ relation is recovered when $u^\star=u_\infty$ and $w_m=0$ almost everywhere on $[0,T]$.

\begin{theorem}[Observer-certified shared-input safety filter]
\label{thm:main-result}
Consider the uncertain plant \eqref{eq:uncertain-lti-system} under Assumptions~\ref{ass:model-uncertainty}, \ref{ass:observer-information}, and~\ref{ass:disturbance}.
For the shared scalar-input case, let $\mathcal U=[u_{\min},u_{\max}]$ and suppose \eqref{eq:relative-degree-two-uncertainty} holds for every $j\in\mathcal J$, with the initial condition in $\cC_H$.
Let $u^\star$ be defined by \eqref{eq:scalar-safety-projection} whenever $\mathcal U_S$ is nonempty.
Then the following statements hold.
\begin{enumerate}
\item \emph{Pointwise feasibility:}
At each fixed $(x,\hat d,\rho)$, the simultaneous observer-certified robust HOCBF and actuator constraints admit a common input if and only if $\mu_S(x,\hat d,\rho)\geq0$, equivalently, $\underline u_S\leq\overline u_S$.

\item \emph{Domain-wise feasibility:}
For every compact connected domain $\Omega$ satisfying Corollary~\ref{cor:domain-feasibility}, $\mu_\Omega\geq0$ guarantees $\mu_S(x,\hat d,\rho)\geq\mu_\Omega$ throughout $\Omega$.
Hence, $\mathcal U_S$ remains nonempty along every trajectory that remains inside the certified domain.

\item \emph{Closed-form filter:}
Whenever $\mu_S\geq0$, the robust safety-admissible set is the interval in \eqref{eq:scalar-safe-input-set}, and the unique solution of \eqref{eq:robust-safety-filter-qp} is the projection \eqref{eq:scalar-safety-projection}.

\item \emph{Robust safety:}
If $\mu_S(x(t),\hat d(t),\rho(t))\geq0$ along the closed-loop trajectory and the regularity conditions of Corollary~\ref{cor:lipschitz-filter} hold, then the projection law admits a locally Lipschitz realization along the feasible trajectory.
The extended closed loop has a unique forward-complete Carathéodory solution, $\norm{d(t)-\hat d(t)}\leq\rho(t)$ for all $t\geq0$, and $\cC_H$ is forward invariant.

\item \emph{Finite-horizon deviation accounting:}
If Assumption~\ref{ass:isaacs} additionally holds, then \eqref{eq:finite-horizon-baseline-deviation-balance} and \eqref{eq:finite-horizon-baseline-deviation-bound} hold on every finite interval over which $\mu_S(x(t),\hat d(t),\rho(t))\geq0$ is maintained.
\end{enumerate}
If $\mu_S(x,\hat d,\rho)<0$ at any point, no admissible input satisfies all imposed robust HOCBF and actuator constraints at that point.
\end{theorem}

\begin{IEEEproof}
Lemma~\ref{lem:estimation-error-bound} provides the observer-certified disturbance set, and Proposition~\ref{prop:exact-robust-hocbf} gives its exact robust HOCBF counterpart over $\mathcal D(\hat d,\rho)\times\mathcal W(x)$.
Proposition~\ref{prop:exact-scalar-feasibility} establishes statement~1 and the closed-form characterization in statement~3.
Corollary~\ref{cor:domain-feasibility} establishes statement~2.
Under the stated regularity conditions, Corollary~\ref{cor:lipschitz-filter} provides a locally Lipschitz realization of the projection law along the feasible trajectory.
Proposition~\ref{prop:robust-forward-invariance} then establishes existence, uniqueness, forward completeness, and forward invariance in statement~4.
Forward completeness extends the observer certificate to all $t\geq0$.
Under Assumption~\ref{ass:isaacs}, Proposition~\ref{prop:finite-horizon-balance} establishes statement~5.
The final infeasibility statement follows directly from Proposition~\ref{prop:exact-scalar-feasibility}.
\end{IEEEproof}

Active CBF constraints can alter the nominal closed-loop dynamics and introduce boundary equilibria or unstable active-set modes \cite{mousavi2026stability}.
Here, a simpler local case arises when the compensated equilibrium input lies strictly inside $\mathcal U_S$.
The projection is then inactive in a neighborhood of the equilibrium, which gives the local stability result below.

\begin{corollary}[Local exponential stability]
\label{cor:local-stability}
Suppose Assumption~\ref{ass:isaacs} holds, $d(t)\equiv d_0$, and the realized structured model uncertainty satisfies $w_m(x,t)\equiv0$.
Let $x_{\mathrm e}:=X_dd_0$, $u_{\mathrm e}:=U_dd_0$, and $\rho_{\mathrm e}:=\omega_1/\lambda_d+\kappa_G\bar w_m(x_{\mathrm e})$.
Suppose $x_{\mathrm e}\in\operatorname{int}\cC_H$, and let $\underline u_{S,\mathrm e}$ and $\overline u_{S,\mathrm e}$ denote the endpoints in \eqref{eq:scalar-safe-interval-bounds} evaluated at $(x,\hat d,\rho)=(x_{\mathrm e},d_0,\rho_{\mathrm e})$.
If
\begin{equation}
\delta_{\mathrm e}:=\min\left\{u_{\mathrm e}-\underline u_{S,\mathrm e},\overline u_{S,\mathrm e}-u_{\mathrm e}\right\}>0,
\label{eq:strict-interior-feasibility}
\end{equation}
then $(x,\hat d,\rho)=(x_{\mathrm e},d_0,\rho_{\mathrm e})$ is a locally exponentially stable equilibrium of the safety-filtered extended closed loop.
\end{corollary}

\begin{IEEEproof}
The regulator equations \eqref{eq:regulator-equations} give $Ax_{\mathrm e}+Bu_{\mathrm e}+Ed_0=0$, while \eqref{eq:disturbance-compensated-nominal-input} gives $u_{\mathrm{nom}}(x_{\mathrm e},d_0)=u_{\mathrm e}$.
The safe-interval endpoints are locally continuous in $(x,\hat d,\rho)$ because $b_j$ is locally continuous in these variables and $a_j\neq0$ on the domain of interest.
Therefore, Condition~\eqref{eq:strict-interior-feasibility} implies the existence of a neighborhood of $(x_{\mathrm e},d_0,\rho_{\mathrm e})$ in which $u_{\mathrm{nom}}\in\operatorname{int}\mathcal U_S$.
Hence, the projection \eqref{eq:scalar-safety-projection} is locally inactive and $u^\star=u_{\mathrm{nom}}$.

Define $\tilde x:=x-x_{\mathrm e}$, $e_d:=d_0-\hat d$, and $\tilde\rho:=\rho-\rho_{\mathrm e}$.
Since $d=d_0$ and $w_m=0$, the local state and observer-error dynamics satisfy
\begin{equation}
\begin{bmatrix}
\dot{\tilde x}\\
\dot e_d
\end{bmatrix}
=
\begin{bmatrix}
A+BF & -BK_{\hat d}\\
0 & -\lambda_d I
\end{bmatrix}
\begin{bmatrix}
\tilde x\\
e_d
\end{bmatrix}.
\label{eq:local-augmented-stability-dynamics}
\end{equation}
The matrix in \eqref{eq:local-augmented-stability-dynamics} is block upper triangular with Hurwitz diagonal blocks $A+BF$ and $-\lambda_dI$, and is therefore Hurwitz.
Moreover, \eqref{eq:dynamic-estimation-error-bound} gives
\begin{equation}
\dot{\tilde\rho}=-\lambda_d\tilde\rho+\lambda_d\kappa_G\left[\bar w_m(x)-\bar w_m(x_{\mathrm e})\right].
\label{eq:local-certificate-dynamics}
\end{equation}
Since $\bar w_m$ is locally Lipschitz, the exponentially convergent subsystem \eqref{eq:local-augmented-stability-dynamics} provides an exponentially decaying input to \eqref{eq:local-certificate-dynamics}, and therefore $\tilde\rho$ also converges exponentially to zero.
The strict interior condition ensures that a sufficiently small neighborhood remains within the region where the safety projection is inactive.
Hence, $(x_{\mathrm e},d_0,\rho_{\mathrm e})$ is locally exponentially stable.
\end{IEEEproof}

\begin{remark}[Scope of the guarantees]
\label{rem:scope-of-guarantees}
The pointwise feasibility characterization in Theorem~\ref{thm:main-result} does not by itself establish recursive feasibility.
Corollary~\ref{cor:domain-feasibility} provides an a priori sufficient condition for nonempty robust safety-admissible input sets over a prescribed compact connected domain.
The domain-wise certificate does not establish that the closed-loop trajectory remains inside that domain.
Corollary~\ref{cor:local-stability} establishes local exponential stability only for a constant disturbance, vanishing realized model uncertainty, and a compensated equilibrium lying strictly inside the certified safe-input interval.
No global stability claim is made for trajectories with persistent safety-filter intervention or nonzero realized model uncertainty, and preservation of the original $H_\infty$ attenuation level is not asserted.
Robust forward invariance otherwise remains conditional on maintenance of $\mu_S\geq0$ along the trajectory.
\end{remark}

\section{Numerical Validation}
\label{sec:twbr}

We evaluate the analytical results of Section~III on the two-wheeled balancing robot (TWBR) model of \cite{fikri2021}.
\subsection{Underactuated Model Characterization}
\label{subsec:twbr-underactuation}

The TWBR has configuration $q=[p\ \theta]^\trans$ and a single wheel-torque input $\tau$, so $\operatorname{rank}B_q(q)=1<2=n_q$ and the underactuation condition of Section~II holds.
Linearization around the equilibrium, with $x=[p\ v\ \theta\ \omega]^\trans$, $v=\dot p$, and $\omega=\dot\theta$, gives $\dot x=Ax+B\tau$ with
\begin{equation}
\begin{aligned}
    A&=\begin{bmatrix}
        0&1&0&0\\
        0&-0.8122&-1.5442&0\\
        0&0&0&1\\
        0&3.1481&44.0083&0
    \end{bmatrix},\\
    B&=\begin{bmatrix}
        0\\
        5.6351\\
        0\\
        -60.6013
    \end{bmatrix}.
\end{aligned}
\label{eq:twbr-nominal-matrices}
\end{equation}
Thus, $\operatorname{rank}B=1$, preserving the local underactuation structure.

For $\sigma\in\{-1,1\}$, define the position and body-pitch barriers as $h_{p,\sigma}(x):=p_{\max}-\sigma p$ and $h_{\theta,\sigma}(x):=\theta_{\max}-\sigma\theta$.
All simulations use $x(0)=0$ over an $8$ s horizon with $p_{\max}=0.08$ m and $\theta_{\max}=0.25$ rad.
For the numerical study, we specialize the general HOCBF functions to $\alpha_{j,k}(s)=\lambda_{j,k}s$.
With the common first-stage HOCBF gain $\lambda_{j,1}=\lambda_1$, the auxiliary functions are $\phi_{p,\sigma}=-\sigma v+\lambda_1(p_{\max}-\sigma p)$ and $\phi_{\theta,\sigma}=-\sigma\omega+\lambda_1(\theta_{\max}-\sigma\theta)$.
Both barriers have relative degree two, with input coefficients
\begin{equation}
\begin{aligned}
a_{p,\sigma}&=-5.6351\sigma,\\
a_{\theta,\sigma}&=60.6013\sigma.
\end{aligned}
\label{eq:twbr-input-coefficients}
\end{equation}
Their opposite signs allow the position and body-pitch HOCBF constraints to impose competing lower and upper bounds on the same wheel torque, providing the shared-actuation structure analyzed in Section~III.

\subsection{Uncertainty and Comparison Setup}
\label{subsec:twbr-comparison-setup}

The uncertain TWBR dynamics are modeled as $\dot x=Ax+B\tau+Ed+Gw_m$, where $E=[0\ -0.5320\ 0\ 5.1629]^\trans$ and
\begin{equation}
G=\begin{bmatrix}
0&0\\
1&0\\
0&0\\
0&1
\end{bmatrix}.
\label{eq:twbr-uncertainty-distribution}
\end{equation}
The choice \eqref{eq:twbr-uncertainty-distribution} preserves the relative degree of both barriers and gives $g_{p,\sigma}=[-\sigma\ 0]$, $g_{\theta,\sigma}=[0\ -\sigma]$, $c_{p,\sigma}=0.5320\sigma$, and $c_{\theta,\sigma}=-5.1629\sigma$.
Hence, $\norm{g_{p,\sigma}}=\norm{g_{\theta,\sigma}}=1$ and $\kappa_G=\norm{E^\dagger G}=0.19267$.
We parameterize these terms as $w_v=\delta_{22}v+\delta_{23}\theta+r_v$ and $w_\omega=\delta_{42}v+\delta_{43}\theta+r_\omega$, where $\abs{\delta_{ij}}\leq\epsilon_A\abs{A_{ij}}$, $\abs{r_v}\leq\bar r_v$, and $\abs{r_\omega}\leq\bar r_\omega$.
A certified bound satisfying Assumption~\ref{ass:model-uncertainty} is
\begin{equation}
\begin{aligned}
\bar w_v&=\epsilon_A(0.8122\abs{v}+1.5442\abs{\theta})+\bar r_v,\\
\bar w_\omega&=\epsilon_A(3.1481\abs{v}+44.0083\abs{\theta})+\bar r_\omega,\\
\bar w_m(x)&=\sqrt{\bar w_v^2+\bar w_\omega^2}.
\end{aligned}
\label{eq:twbr-model-uncertainty-bound}
\end{equation}

Solving \eqref{eq:gare} with $Q=\operatorname{diag}(1.613^2,0,1.779^2,0)$, $R=0.733^2$, and $\gamma=0.086$ gives $F=[2.994763\ 3.486833\ 6.553652\ 0.948437]$.
With $C_r=[1\ 0\ 0\ 0]$, \eqref{eq:regulator-equations} gives $X_d=[0\ 0\ 0.0203771\ 0]^\trans$, $U_d=0.0999922$, and $K_{\hat d}=-0.0335519$.
The observer gain is set to $\lambda_d=120$, corresponding to the nominal estimation-error decay time constant $1/\lambda_d=8.33$ ms.
Under the linear HOCBF specialization used in the numerical study, the gains are fixed uniformly at $\lambda_{j,1}=\lambda_{j,2}=8$ for all constraints to avoid constraint-specific tuning.
The disturbance-rate bound $\omega_1$ is selected offline from the prescribed bound on $\abs{\dot d}$ and does not require online measurement of the disturbance derivative.

For the uncertainty tests, we set $\epsilon_A=3.5\times10^{-3}$ and $\bar r_v=\bar r_\omega=0$.
All $2^4=16$ sign combinations of $(\delta_{22},\delta_{23},\delta_{42},\delta_{43})$ satisfying $\abs{\delta_{ij}}=\epsilon_A\abs{A_{ij}}$ are evaluated.
Thus, every compared method is tested on the same plant-uncertainty realizations, state constraints, disturbance, and actuator limits.

We compare five safety-filter configurations.
B1 is a fixed-bound robust HOCBF using the scenario-wise bounds $\abs{d}\leq10$ N and $\abs{d}\leq12$ N for the smooth-doublet and temporary-step cases, respectively, together with the same $\bar w_m(x)$.
B2 and B3 apply the higher-relative-degree DOB-HOCBF of \cite{wangxu2023} with $(\alpha,\nu)=(120,80)$ and the uncertainty-estimation HOCBF of \cite{dasburdick2025} with $\Lambda=120I$, respectively, to the aggregate uncertainty $Ed+Gw_m$.
B4 is the feasibility-aware optimal-decay HOCBF adaptation of \cite{zeng2021optimal} with $p_\omega=100$.
To assess tuning sensitivity under the actuator-stress case, B2 is additionally evaluated over the eight admissible pairs from $\alpha\in\{60,120,180\}$ and $\nu\in\{40,80,120\}$, B3 over $\lambda\in\{60,120,180,240\}$, and B4 over $p_\omega\in\{10,100,1000\}$.
All sensitivity tests use the same $16$ uncertainty corners with fixed $H_\infty$ and HOCBF gains.
Here, $\mu$ denotes the method-specific feasible-input width for B1 and B2, while $\mu_S$ denotes the proposed margin.

Let $s_{t_0}(t):=\frac12[1+\tanh(10(t-t_0))]$.
The smooth-doublet and smoothed temporary-step disturbance profiles are
\begin{equation}
\begin{aligned}
d_{\mathrm{dbl}}(t)&=10[s_{0.5}(t)-s_{1.5}(t)-s_{2.0}(t)+s_{3.0}(t)],\\
d_{\mathrm{step}}(t)&=7s_{0.1}(t)+5s_{1.0}(t)-12s_{2.0}(t).
\end{aligned}
\label{eq:twbr-comparison-disturbances}
\end{equation}
Since $\dot s_{t_0}(t)=5\operatorname{sech}^2(10(t-t_0))$, the selected observer parameters are $(\bar e_0,\omega_1)=(0.001,55)$ and $(0.85,65)$ for the smooth-doublet and temporary-step cases, respectively, with $\omega_1$ exceeding their nominal peak rates of approximately $50$ N/s and $60$ N/s.
For the upright equilibrium $d_0=0$, $x_{\mathrm e}=0$, and $u_{\mathrm e}=0$ with the conservative stress-test value $\omega_1=65$, $\rho_{\mathrm e}=0.54167$ N and $\mathcal U_{S,\mathrm e}=[-0.21787,0.21787]$ N$\cdot$m.
Thus, $\delta_{\mathrm e}=0.21787$ N$\cdot$m satisfies \eqref{eq:strict-interior-feasibility}, while the eigenvalues $-25.2121$, $-11.0896$, $-1.1692\pm0.5356\mathrm{i}$ and observer mode $-120$ verify the remaining local-stability conditions.

The $16$-corner sweeps use fixed-step RK4 with $\Delta t=10^{-3}$ s; Fig.~\ref{fig:observer-error-bound} uses the worst-margin actuator-stress trajectory.
At all $8001$ samples of this trajectory, the closed-form projection \eqref{eq:scalar-safety-projection} agrees with an SLSQP numerical solution of \eqref{eq:robust-safety-filter-qp} to $4.44\times10^{-16}$ N$\cdot$m.

\subsection{Comparative Validation}
\label{subsec:twbr-results}

On one core of an AMD EPYC 9V74, five repetitions over the same $8001$ stress-trajectory samples give mean update times of $11.03~\mu$s for the closed-form projection and $266.75~\mu$s for the SLSQP numerical realization of \eqref{eq:robust-safety-filter-qp}, corresponding to a $24.2\times$ implementation-level reduction.
These timing results are implementation benchmarks rather than certified worst-case execution times.
Across all $16$ proposed stress-test trajectories, $\epsilon_\Gamma:=\max_{j,t}\max\{-\Gamma_j^{\mathrm{rob}}(t),0\}=7.11\times10^{-15}$, indicating robust-HOCBF satisfaction to numerical roundoff.

Table~\ref{tab:architecture-ablation} isolates the baseline $H_\infty$ policy with actuator saturation, its observer-based disturbance-compensated counterpart, and the complete proposed safety filter under identical plant realizations and actuator limits.

\begin{table}[t]
\centering
\caption{Architecture ablation over all $16$ structured-uncertainty corners.}
\label{tab:architecture-ablation}
\scriptsize
\setlength{\tabcolsep}{3.5pt}
\begin{tabular}{lccc}
\toprule
& \multicolumn{2}{c}{Smooth doublet}
& Stress test\\
Configuration & Safe & $(\max\abs{p},\max\abs{\theta})$ & Safe\\
\midrule
$H_\infty$ & $16/16$ & $(0.06550,\,0.20460)$ & $0/16$\\
$H_\infty$ + compensation & $16/16$ & $(0.05196,\,0.24533)$ & $0/16$\\
Full proposed & $16/16$ & $(0.05615,\,0.21109)$ & $\mathbf{16/16}$\\
\bottomrule
\end{tabular}
\end{table}

Under the smooth doublet, compensation reduces the position excursion but moves the pitch closer to its limit, whereas the full filter retains both constraints.
Under the stress test, the saturated baseline and compensated controllers fail in all $16$ corners at earliest times $1.344$ s and $1.217$ s, respectively, while the full filter remains safe.
Table~\ref{tab:robust-comparison} reports the corresponding safety-filter comparison.
\begin{table*}[t]
\centering
\caption{Safety-filter comparison over all 16 structured-uncertainty corners. ``Completed'' denotes corners remaining feasible and HOCBF-safe over the full horizon, while ``loss'' denotes the earliest method-specific feasibility or safety loss among the tested corners.}
\label{tab:robust-comparison}
\scriptsize
\setlength{\tabcolsep}{5pt}
\begin{tabular}{lcc|cc}
\toprule
& \multicolumn{2}{c|}{Smooth doublet, $\tau_{\max}=3$ N$\cdot$m}
& \multicolumn{2}{c}{Smoothed temporary step, $\tau_{\max}=1.18$ N$\cdot$m}\\
Method & Completed & Feasibility indicator & Completed & Feasibility indicator\\
\midrule
B1 Fixed robust HOCBF & $0/16$ & infeasible at $t=0$ & $0/16$ & $\mu(0)=-1.517$\\
B2 Wang--Xu DOB-HOCBF & $16/16$ & $\min\mu>0$ & $0/16$ & loss at $1.202$ s\\
B3 Da\c{s}--Burdick UE-HOCBF & $0/16$ & loss at $0.989$ s & $0/16$ & loss at $0.780$ s\\
B4 Feasibility-aware HOCBF adaptation & $0/16$ & loss during horizon & $0/16$ & loss at $1.077$ s\\
Proposed & $\mathbf{16/16}$ & $\min\mu_S=0.448$ & $\mathbf{16/16}$ & $\min\mu_S=0.007267$\\
\bottomrule
\end{tabular}
\end{table*}

Under the smooth doublet, only B2 and the proposed filter complete all $16$ corners, while the proposed filter retains $\min\mu_S=0.448$ N$\cdot$m.
Under the temporary-step stress test, only the proposed filter completes all corners, with $\max\abs{p}=0.07403$ m, $\max\abs{\theta}=0.20543$ rad, $\max\abs{\tau}=1.17273$ N$\cdot$m, and $\min\mu_S=0.007267$ N$\cdot$m.
The stress-case sensitivity sweep does not change the B2--B4 completion counts, with their best earliest-loss times improving only to $1.206$ s at $(\alpha,\nu)=(180,120)$, $0.802$ s at $\lambda=240$, and $1.078$ s at $p_\omega=10$, respectively.
These results apply to the stated TWBR adaptations and uncertainty bounds rather than constituting general comparisons among the underlying methods.

Figure~\ref{fig:exact-feasibility-validation} examines the worst-margin temporary-step trajectory at $t^\star=1.286$ s.
Figure~\ref{fig:exact-feasibility-validation}(a) freezes $(v,\omega,\hat d,\rho)$ at their realized values and varies $(p,\theta)$, showing the physical safe set, HOCBF-set slice, and feasibility domains with and without the actuator bound $\abs{\tau}\leq1.18$ N$\cdot$m.
The realized state lies inside but near the bounded-input feasibility boundary.
At this state, the unbounded robust interval is $[1.172733,1.606408]$ N$\cdot$m, while intersection with the actuator set gives $\mathcal U_S(t^\star)=[1.172733,1.18]$ N$\cdot$m and $\mu_S(t^\star)=0.007267$ N$\cdot$m.
Replacing the norm-ball model-uncertainty bound by its componentwise structured TWBR bound changes this margin only to $0.0072674$ N$\cdot$m, whereas an oracle evaluation using the realized disturbance error and model uncertainty gives $0.05417$ N$\cdot$m.
Since $u_{\mathrm{nom}}(t^\star)=0.767012$ N$\cdot$m, the projection gives $u^\star(t^\star)=\underline u_S(t^\star)=1.172733$ N$\cdot$m, with the positive-pitch constraint and actuator limit determining the minimum-margin interval.
The analytical condition $\mu_S\geq0$ agrees with an independent numerical feasibility test at all $320$ sampled points of the frozen slice.

\begin{figure*}[t]
\centering
\includegraphics[width=0.98\textwidth]{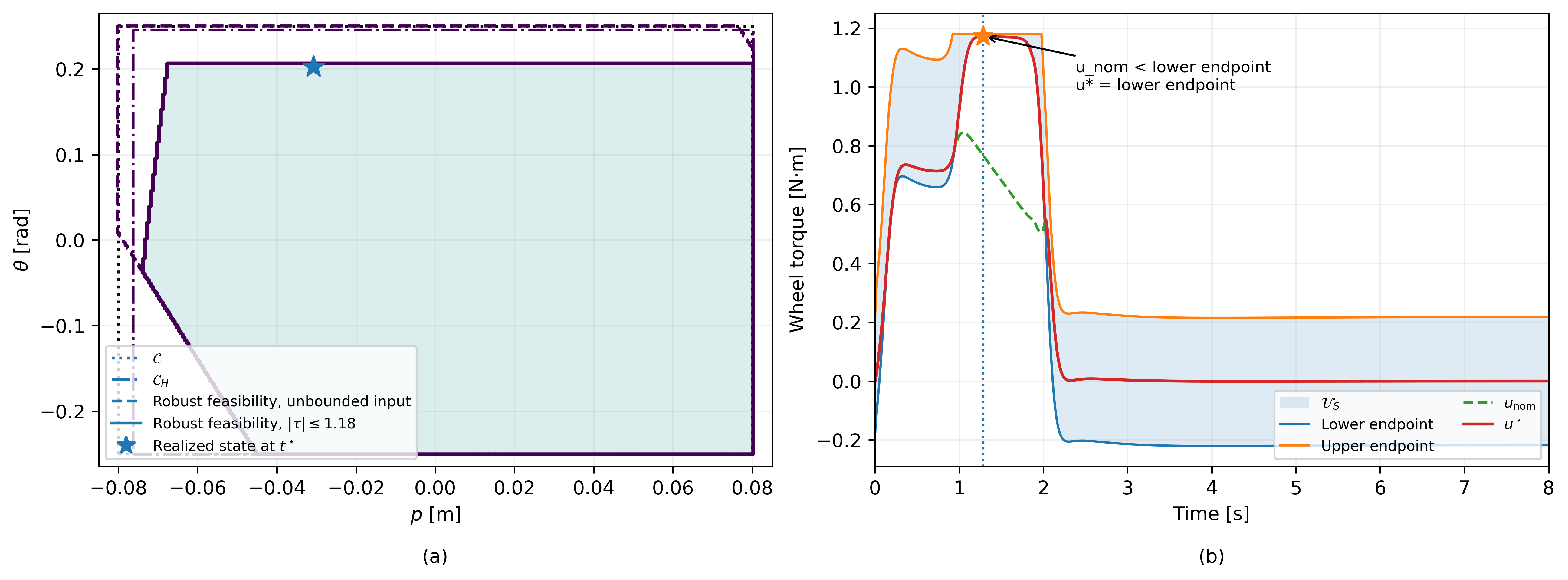}
\caption{Exact shared-input feasibility for the worst-margin stress trajectory.
(a) Frozen $(p,\theta)$ feasibility slice at $t^\star=1.286$ s, with the realized state marked by the star.
(b) Exact safe-input interval, nominal input, and projected input, with the positive-pitch constraint and actuator limit determining the minimum-margin endpoints.}
\label{fig:exact-feasibility-validation}
\end{figure*}

The positive- and negative-pitch constraints determine the lower and upper endpoints over $85.76\%$ of the worst-margin trajectory.
During the actuator-limited transient, the upper endpoint switches to $u_{\max}$, while the lower endpoint transitions from the positive-pitch to the negative-position constraint.
The projection is active at $1143$ of $8001$ samples and always satisfies $u_{\mathrm{nom}}<\underline u_S$, so the filtered input is projected onto the lower endpoint.

\begin{figure}[t]
    \centering
    \includegraphics[width=\columnwidth]{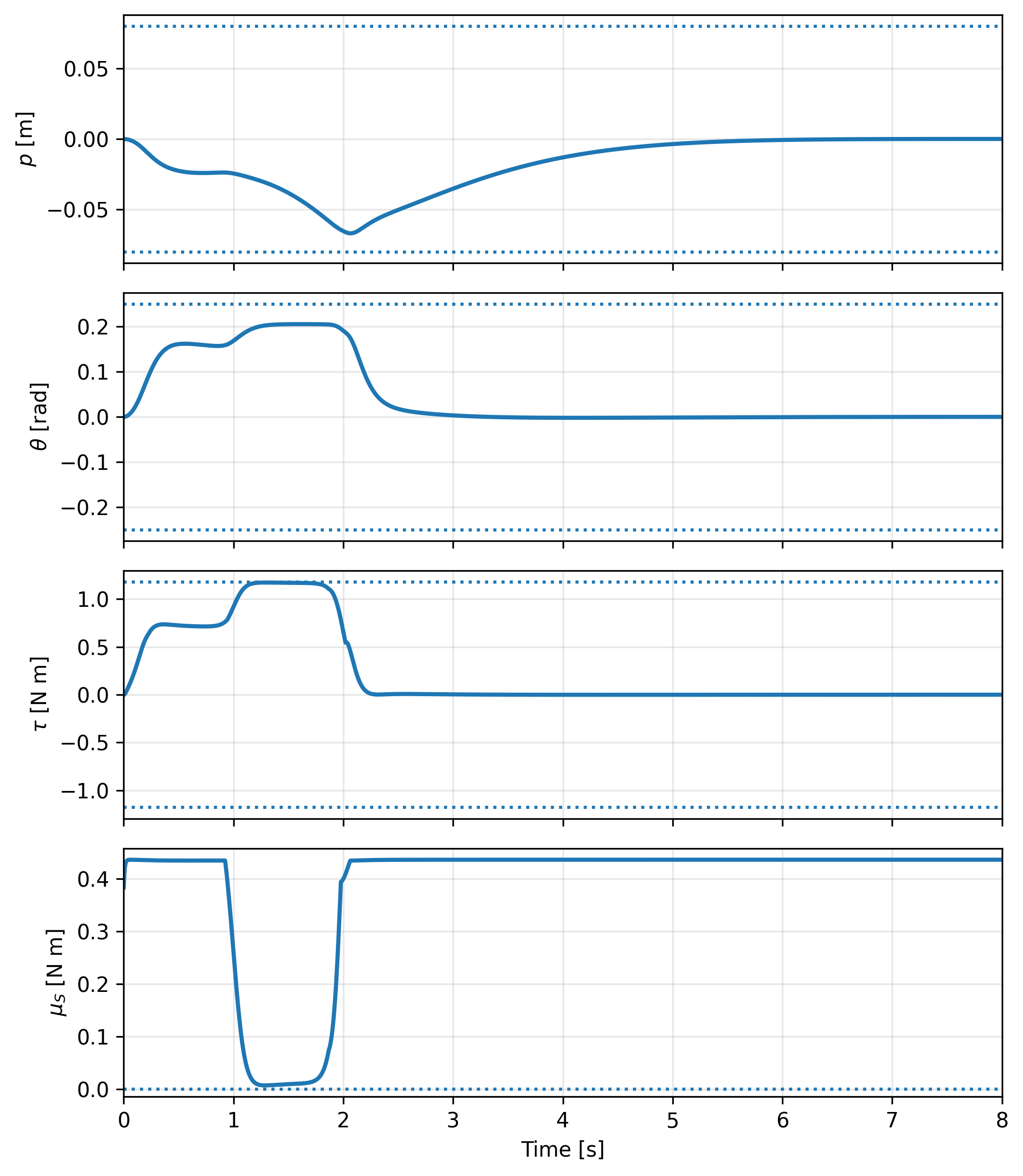}
    \caption{Proposed-filter trajectory at the worst-margin uncertainty corner of the temporary-step actuator-stress test. The panels show position, pitch angle, wheel torque, and the exact safe-input interval margin; dotted lines denote the corresponding constraint boundaries and zero feasibility margin.}
    \label{fig:actuator-stress-trajectories}
\end{figure}

Figure~\ref{fig:actuator-stress-trajectories} shows the worst-margin corner, where the state constraints remain satisfied while the torque approaches its bound and $\mu_S(t)$ contracts to $0.007267$ N$\cdot$m without becoming negative.

\begin{figure}[t]
    \centering
    \includegraphics[width=\columnwidth]{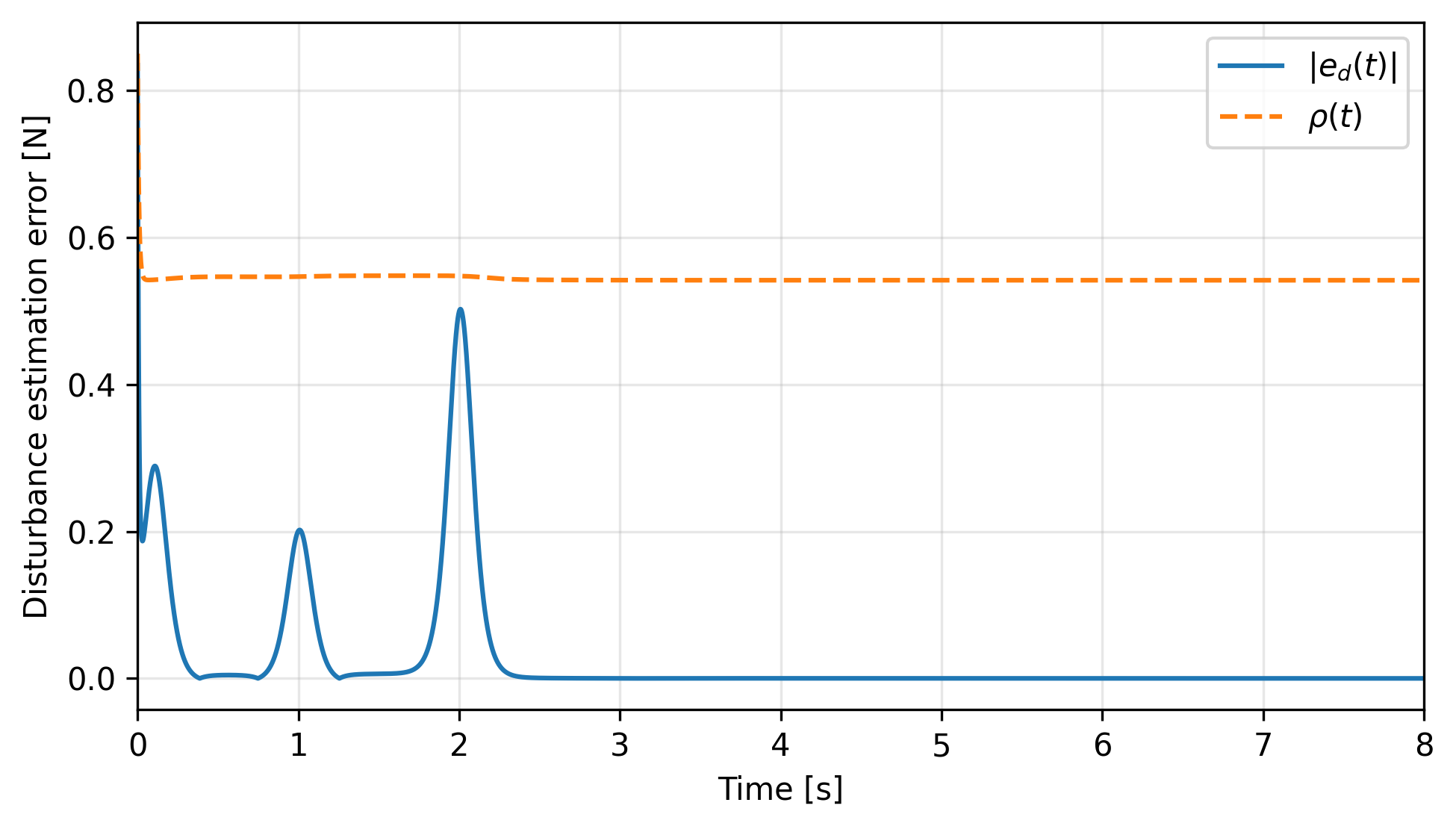}
    \caption{Disturbance-estimation error magnitude and certified radius at the worst-margin uncertainty corner of the temporary-step actuator-stress test.}
    \label{fig:observer-error-bound}
\end{figure}

Figure~\ref{fig:observer-error-bound} verifies that $\abs{e_d(t)}\leq\rho(t)$ with a minimum certificate slack of $0.01558$ N, while the largest absolute finite-horizon balance residual across the $16$ stress-test trajectories is $4.05\times10^{-7}$.

\subsection{Out-of-Model Nonlinear Validation}

We also propagate the closed loop with the nonlinear TWBR mechanics using the controller, observer, HOCBF gains, and uncertainty bounds designed from the linear model.
The physical parameters are $m_B=1.25$ kg, $m_w=0.40$ kg, $I_B=0.0178$ kg$\cdot$m$^2$, $I_w=0.004$ kg$\cdot$m$^2$, $r=0.10$ m, $b_w=0.01$ N$\cdot$m$\cdot$s/rad, $l=0.08$ m, $l_d=0.08$ m, and $g=9.81$ m/s$^2$ \cite{fikri2021}.
Define $A_1:=m_B+2m_w+2I_w/r^2$, $A_2:=2b_w/r^2$, $A_3=A_5:=m_Bl$, $A_4:=1/r$, $A_6:=m_Bl^2+I_B$, and $A_8:=l_d$.
The nonlinear acceleration dynamics are
\begin{equation}
\begin{aligned}
A_1\dot v+A_3\cos\theta\,\dot\omega
&=A_4\tau-d-A_2v\\
&\quad+A_3\sin\theta\,\omega^2,\\
A_5\cos\theta\,\dot v+A_6\dot\omega
&=-\tau+A_8d+m_Bgl\sin\theta.
\end{aligned}
\label{eq:nonlinear-twbr-validation}
\end{equation}
Linearization of \eqref{eq:nonlinear-twbr-validation} about $(\theta,\omega)=(0,0)$ recovers the nominal matrices in \eqref{eq:twbr-nominal-matrices} and the disturbance distribution $E$ used in the linear-model study up to the reported numerical rounding.
For each nonlinear simulation, the same corner-dependent acceleration perturbations $w_v=\delta_{22}v+\delta_{23}\theta$ and $w_\omega=\delta_{42}v+\delta_{43}\theta$ defined in Section~IV-B are retained after evaluating the nonlinear nominal accelerations, so that the $16$ corner realizations remain identical to those used in the corresponding linear comparison.
The additional nonlinear mechanical terms generated by \eqref{eq:nonlinear-twbr-validation} are not included in the certified bound $\bar w_m(x)$ and are therefore treated only through out-of-model empirical validation.

\begin{table}[t]
\centering
\caption{Out-of-model nonlinear validation over the $16$ uncertainty corners.}
\label{tab:linear-nonlinear-validation}
\scriptsize
\setlength{\tabcolsep}{3pt}
\begin{tabular}{lcccc}
\toprule
Case & Complete & $(\max\abs{p},\max\abs{\theta})$ & $\max\abs{\tau}$ & $\min\mu_S$\\
& & [m, rad] & [N$\cdot$m] & [N$\cdot$m]\\
\midrule
Smooth & $16/16$ & $(0.05708,0.21106)$ & $1.00372$ & $0.44776$\\
Stress & $8/16$ & $(0.06984,0.20540)$ & $1.17160$ & $0.008401$\\
\bottomrule
\end{tabular}
\end{table}

Without retuning, the filter completes all $16$ nonlinear smooth-doublet cases and $8$ of $16$ actuator-stress cases, with the stress-case extrema in Table~\ref{tab:linear-nonlinear-validation} reported over the completed runs.
The remaining cases lose feasibility between $1.787$ s and $2.567$ s, all at the negative extreme of $\delta_{43}$.
Because the nonlinear remainder lies outside the certified envelope $\mathcal D(\hat d,\rho)\times\mathcal W(x)$, these results are empirical and do not extend Theorem~\ref{thm:main-result}.
The contrast with the linear stress margin $\min\mu_S=0.007267$ N$\cdot$m shows that the small remaining feasibility reserve can be exhausted by the out-of-model nonlinear dynamics.

\section{Conclusion}

An observer-certified robust safety filter was developed for input-constrained linear systems with bounded-rate disturbances and structured model uncertainty.
The robust constraints admit a polyhedral feasibility characterization for general polytopic inputs and an exact interval, signed feasibility reserve, and closed-form projection for the shared scalar-input case.
A domain-wise actuator-authority certificate and strict-interior condition further characterize feasible operation and local stability, while the finite-horizon analysis quantifies deviation from the unconstrained $H_\infty$ reference without claiming preservation of its original attenuation level.
In the linear-model TWBR study, the filter remains feasible and safe across all $16$ uncertainty corners, including the $\tau_{\max}=1.18$ N$\cdot$m stress case with $\min\mu_S=0.007267$ N$\cdot$m.
The nonlinear validation requires no retuning and remains feasible in the smooth-doublet case, while the tighter actuator-stress case exposes feasibility loss outside the certified uncertainty model.
Closed-form multi-input filtering, recursive-feasibility certification, and uncertainty outside known distribution channels remain open.

\bibliographystyle{IEEEtran}
\bibliography{cbf_observer_references}

\begin{IEEEbiography}[{\includegraphics[
width=1in,
height=1.25in,
clip,
keepaspectratio]{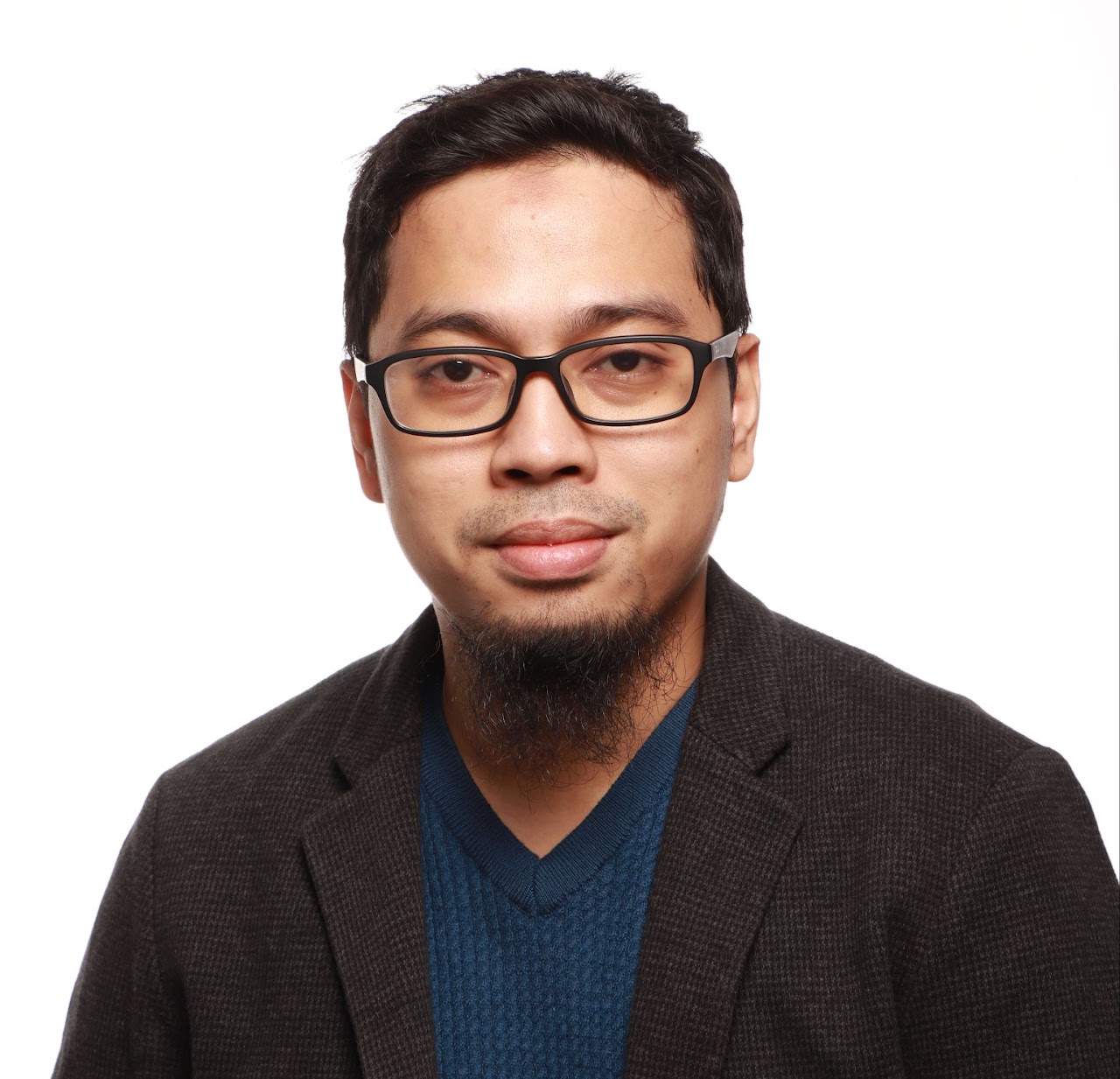}}]{Muhamad Rausyan Fikri}
is with the Faculty of Engineering and Natural Sciences, Tampere University, Tampere, Finland.
His research interests include robust and safety-critical control, nonlinear control, autonomous systems, and multi-agent robotic systems, with particular emphasis on disturbance-resilient control, distributed algorithms, control barrier functions, persistent task allocation, and multi-agent coordination.
\end{IEEEbiography}

\end{document}